\documentclass[journal]{IEEEtran}

\usepackage{amsmath,amssymb}
\usepackage{amsthm}
\usepackage{xcolor}
\usepackage[ruled, linesnumbered]{algorithm2e}
\usepackage{lipsum}	
\usepackage{graphicx}
\usepackage{pifont}
\usepackage{paralist}
\usepackage{float}
\usepackage{hyperref}
\usepackage{pifont}
\usepackage{authblk}
\usepackage{caption}
\usepackage{multirow}
\usepackage{subcaption}
\usepackage[numbers]{natbib}
\usepackage[switch]{lineno}

\newtheorem{theorem}{Theorem}
\newtheorem{lemma}{Lemma}

\theoremstyle{definition}
\newtheorem{defn}{Definition}

\newtheorem{prop}{Proposition}

\theoremstyle{remark}

\newcommand{\ignore}[1]{}
\newcommand{\user}{{{\sf User}}\xspace}
\newcommand{\admin}{{{\sf Admin}}\xspace}
\newcommand{\analyst}{{{\sf Analyst}}\xspace}
\renewcommand{\epsilon}{\varepsilon}

\title{Detecting Anomalous LAN Activities under Differential Privacy$^*$\thanks{$^*$The definitive Version of Record
		was published in \textit{Security and Communication Networks, vol. 2022, Apr. 2022}. \url{https://doi.org/10.1155/2022/1403200}}}

\author[1]{Norrathep Rattanavipanon}
\author[2]{Donlapark Ponnoprat}
\author[3]{Hideya Ochiai}
\author[1]{Kuljaree Tantayakul}
\author[4]{\\Touchai Angchuan}
\author[4]{Sinchai Kamolphiwong}
\affil[1]{College of Computing, Prince of Songkla University, Phuket, 83120, Thailand \authorcr {\small\tt \{norrathep.r, kuljaree.t\}@phuket.psu.ac.th}}
\affil[2]{Data Science Research Center, Department of Statistics, Faculty of Science,\protect\\ Chiang Mai University, Chiang Mai, 50200, Thailand \authorcr {\small\tt donlapark.p@cmu.ac.th}} 
\affil[3]{Graduate School of Information Science and Technology, The University of Tokyo, Tokyo, 113-8656, Japan \authorcr {\small\tt ochiai@elab.ic.i.u-tokyo.ac.jp}} 
\affil[4]{Faculty of Engineering, Prince of Songkla University, Songkhla, 90110, Thailand \authorcr {\small\tt \{touch, ksinchai\}@coe.psu.ac.th}}

\begin{document}
	\maketitle
	
	\begin{abstract}
		Anomaly detection has emerged as a popular technique for detecting malicious activities in local area networks (LANs).
		Various aspects of LAN anomaly detection have been widely studied.
		Nonetheless, the privacy concern about individual users or their relationship in LAN has not been thoroughly explored in the prior work.
		In some realistic cases, the anomaly detection analysis needs to be carried out by an external party, located outside the LAN.
		Thus, it is important for the LAN admin to release LAN data to this party in a private way in order to protect privacy of LAN users;
		at the same time, the released data must also preserve the utility of being able to detect anomalies.
		This paper investigates the possibility of privately releasing ARP data that can later be used to identify anomalies in LAN.
		We present four approaches, namely na\"ive, histogram-based, na\"ive-$\delta$ and histogram-based-$\delta$, and show that they satisfy different levels of differential privacy -- a rigorous and provable notion for quantifying privacy loss in a system.
		Our real-world experimental results confirm practical feasibility of our approaches. With a proper privacy budget, all of our approaches preserve more than 75\% utility of detecting anomalies in the released data.
	\end{abstract}

	\section{Introduction}
	
	Security of local area networks (LANs) has been getting more attention in the last few decades.
	Traditional LAN defense mechanisms based on a firewall are no longer effective in preventing malware infection since malware can simply circumvent the firewall or infect the network through other means~\cite{vuln, yan2008revealing}.
	A prominent example is the recent emergence of ransomware that can infect LAN devices via phishing attacks; these attacks remain effective even if the LAN's firewall is active and configured correctly~\cite{ransom1, ransom2}. In addition, with the rise of the Internet-of-things (IoT), the so-called ``smart" devices have become widely popular and, at the same time, are also extremely vulnerable to malware attacks~\cite{iotmalware}.
	These devices may be infected from the outside world and introduce malware to the LAN.
	
	To overcome this challenge, several anomaly detection techniques have been proposed to detect malicious activities in LAN. 
	Among those, techniques based on the Address Resolution Protocol (ARP) are shown to be promising in detecting anomalous activities in LAN without requiring a change to existing devices~\cite{matsufuji2019arp, yasami2007arp}, making it suitable to the current IoT networks.
	
	Despite this success, there still remains a severe privacy concern to LAN users, which has not been thoroughly explored in the previous work. 
	Often times, the anomaly detection must be performed by an entity outside LAN~\cite{ren2019time, mobilio2019anomaly, yao2017anomaly} or third-party software~\cite{microsoftanomaly, tibco}.
	Thus, it is equally important to ensure privacy of the data exposed to this external and potentially malicious entity.
	For instance, a LAN admin in an enterprise may choose to outsource an anomaly detection analysis to an external widely-popular service, e.g., Microsoft's Anomaly Detector~\cite{microsoftanomaly}, or the admin simply wants to release some features of network data for transparency or academic purposes. 
	In either case, it would require the LAN admin to output network data (which is an input to the anomaly detection algorithm) to an untrusted party.
	Doing so may lead to having such party learn privacy-sensitive information about the LAN users. 
	For example, it may directly disclose personally identifiable information (PII), e.g., IP/MAC addresses, which can be used to uncover the identity of LAN users. It may also cause an indirect information leakage by revealing information about access patterns (e.g., the time of the day that a specific user is online) or relationship between users~\cite{hu2016relationship}.
	
	While it is possible to simply erase all users' sensitive information from the output data, this kind of technique does not provide strong and provable privacy guarantees. A motivated adversary may still be able to deanonymize users through other means, e.g., performing a side-channel analysis~\cite{srivatsa2012deanonymizing}
	or correlating the remaining network traces with the physical world data ~\cite{takbiri2018privacy}.
	Therefore, there is a need for a technique with \emph{rigorous} privacy guarantees, while preserving the utility of detecting anomalies in the LAN environment.~\\
	
	\noindent\textbf{Contributions:}
	To this end, the goal of this paper is to investigate the possibility of privately publishing ARP data that can later be used to identify anomalies in LAN. Our work presents the following contributions:
	
	\begin{itemize}
		\item \textbf{Privacy Notions for ARP Publication.} We identify four concrete privacy notions in the context of ARP-data publication.
		Each notion is defined over a different type of information that needs to be privacy-protected as well as the probability that this protection holds. 
		Specifically, they are derived from the widely-known \emph{differential privacy}~\cite{Dwork2006} notion, which allows us to mathematically prove whether a specific algorithm adheres to any of these notions.
		We argue that this is a necessary and essential step towards designing, implementing and deploying any privacy-preserving approach into the real world. Without it, it is doubtful whether any meaningful guarantee can be obtained from our approaches.
		
		\item \textbf{Releasing ARP for Anomaly Detection with Various Degrees of Privacy.} We present four approaches capable of privately releasing ARP data that still preserves the utility of detecting LAN anomalies.
		Our approach provides a wide range of privacy-preserving degrees, making them suitable to different scenarios:
		\begin{itemize}
			\item The first approach requires small additive perturbations to the input ARP data in exchange for privacy protection of user relationship.
			
			\item The second approach perturbs the input data by a relatively higher amount but it can attain a stronger privacy protection guarantee for each individual LAN device/user.
			
			\item The third and fourth are variants of the first two approaches that require even smaller data perturbations; however, they sacrifice some small probability that the privacy guarantee will not hold, making them an appropriate option for scenarios where data utility needs to be maximized.  
		\end{itemize}
		
		\item \textbf{Practicality via Real-world Deployment.} We demonstrate practicality of our approaches by implementing and deploying them as part of a large-scale real-world project, called ASEAN-Wide Cyber-Security Research Testbed Project~\footnote{\url{https://www.nict.go.jp/en/asean_ivo/ASEAN_IVO_2020_Project03.html}}. Overall, the aim of this project is three-fold: (1) to capture network data from multiple LANs across the ASEAN region, (2) to determine malware behaviors based on the captured data and (3) to make the captured data sharable in the public domain.
		Our work fits perfectly in this project as it fulfills the third goal by providing a privacy-preserving mechanism for releasing captured ARP data.
		
		\item \textbf{Evaluation on Real-world Dataset.}
		We evaluate our approaches on a real-world ARP dataset captured from 3 LANs over 30 weeks. 
		The experimental result shows feasibility of our approaches as they introduce only low error values ($<10$ in the root-mean-square error) to the original data. In addition, we assess utility of the released data by testing it on the existing LAN anomaly detector~\cite{matsufuji2019arp}. The result is promising as our approaches can achieve $75\%$ anomaly detection rate. 
		
	\end{itemize}

	\noindent\textbf{Organization:} The rest of the paper is organized as follows: Section~\ref{sec:related} overviews existing work related to LAN anomaly detection and differential privacy. The background in Address Resolution Protocol and differential privacy are discussed in Section~\ref{sec:prelim}. Section~\ref{sec:prob} describes the system and adversarial models targeted in this work. Section~\ref{sec:notions} presents privacy notions in the context of releasing ARP data. Section~\ref{sec:arp-eps} and Section~\ref{sec:arp-delta} present four approaches and prove that they satisfy privacy notions defined in the previous section. Experiments are carried out and reported in Section~\ref{sec:exp}. Several issues are discussed in Section~\ref{sec:disc}. Finally, the paper concludes in Section~\ref{sec:conclude}.
	
	\section{Related Work}\label{sec:related}
	
	\noindent\textbf{Differential privacy in anomaly detection.}
	To the best of our knowledge, there has been no prior work that proposes a release mechanism for ARP data with differential privacy guarantees while retaining the utility of anomaly detection in the LAN setting.
	The closest related work can be found in~\cite{mcsherry2010differentially}, where the authors employ PINQ differential privacy framework~\cite{mcsherry2009privacy} to detect network-wide traffic anomalies. The main difference between our work and the work in~\cite{mcsherry2010differentially} lies in the type and magnitude of the released data as well as the privacy guarantee. 
	The work in \cite{mcsherry2010differentially} aims to privately release \emph{link-level traffic volumes of ISP} whose overall value tends to be much larger than noise introduced by any differentially-private release mechanism. On the other hand, our work operates on more restricted input (ARP-degree) which generally contains a much smaller value, making it more noise-sensitive than ISP's traffic volume. Reducing this sensitivity poses a main challenge addressed in this work. Further, the work in~\cite{mcsherry2010differentially} provides \emph{no} privacy protection guarantee for individual network users. Achieving this guarantee is non-trivial, as discussed in Section~\ref{histogram}.
	
	Besides the work in~\cite{mcsherry2010differentially}, several existing work focuses on providing anomaly detection with differential privacy guarantees in non-networked settings, e.g., web browsing~\cite{fan2014monitoring}, social network~\cite{wang2016real}, health care~\cite{dankar2013practicing},
	or syndrome surveillance~\cite{fan2013differentially}. Due to the difference in the target setting, the aforementioned techniques are not directly applicable to our work.~\\ 
	
	
	\noindent\textbf{LAN anomaly detection.}
	There are a number of existing research that aims to detect anomalies in LAN \emph{without} providing privacy protection. 
	Zhang et al.~\cite{zhang2019unveiling} present an approach based on honeypot to detect malicious LAN activities.
	Yeo et al.~\cite{yeo2004framework} propose a framework to monitor a network traffic and detect anomalies in the Wireless LAN (WLAN) environment via the IEEE 802.11 MAC protocol.
	Nonetheless, this approach is specific to wireless LAN and thus cannot be directly applied to the wired LAN setting.
	Our approaches are based on ARP requests, making them suitable for both wired and wireless LAN environments. 
	
	Several prior work focuses on detecting LAN anomalies based on ARP-related data.
	Whyte et al.~\cite{whyte2005arp} propose an anomaly detection approach that distinguishes anomalous activities through statistical analyses of ARP traffic.
	Yasami et al.~\cite{yasami2007arp} propose to model normal ARP traffic behaviors using Hidden Markov Model.
	Farahmand et al.~\cite{farahmand2006multivariate} detect LAN anomalies based on four features: traffic rate, burstiness, dark space and sequential scan.
	Matsufuji et al.~\cite{matsufuji2019arp} present an anomaly detection algorithm based on the degree of destination of ARP requests.
	
	\section{Background}\label{sec:prelim}
	
	\subsection{Address Resolution Protocol (ARP)}
	
	In a nutshell, ARP is a request-response protocol that provides a mapping between dynamic IP addresses and permanent link-layer addresses (also known as MAC addresses), allowing one computer to discover a MAC address of another from its IP address. This protocol is essential in a LAN environment since it enables communication between any two computers within the same sub-network as follows:
	 
	In LAN, when one computer needs to connect with another, it uses ARP to broadcast a request asking for the MAC address associated with the IP address of the destination computer. 
	Therefore, an ARP request contains the requester's IP and MAC addresses as well as the destination's IP address.
	Upon receiving the ARP request, every computer checks whether the received IP address matches with one of its network interfaces. If it does, it unicasts an ARP response back to the requester along with its IP and MAC addresses. 
	At the end of this process, the requester successfully retrieves the destination's MAC address and can use this information to construct Ethernet frames for transmitting subsequent data to the target computer.
	
	Similar to other network protocols, ARP involves using sensitive data that has previously been shown to be directly (e.g., IP address) or indirectly (e.g., traffic volume~\cite{takbiri2018privacy}) linkable to the identity of network users.
	Hence, this privacy concern must be taken into account when designing an approach for releasing ARP data.
	
	\subsection{Differential Privacy (DP)}

    Consider a setting in which there are $n$ users who send individual data to a trusted curator. The curator then applies an algorithm $\mathcal{M}$ and outputs these results to an untrusted party. In a strong notion of privacy, the data of an individual must be kept private from strong adversaries -- even ones who get a hand on the data of the other users. 
    
    The \emph{differential privacy} (DP) is a viewpoint of this notion given in a seminal paper by Dwork, McSherry, Nissim, and Smith~\cite{Dwork2006}. First, we say that two databases $X$ and $X'$ are \emph{neighboring} if they differ by exactly one database entry. The differential privacy is then satisfied if changing $X$ to $X'$ does not change the probability of observing an output of $\mathcal{M}$ by very much. With differential privacy, presence of a single entry will not affect the published output by much. Therefore, outputs from a differentially-private algorithm cannot be used to infer about any single entry from the input dataset.
    \begin{defn}[Differential Privacy]\label{def:dp}
        An algorithm $\mathcal{M}:\mathcal{X}\to\mathcal{Y}$ satisfies $(\epsilon,\delta)$-\emph{differential privacy} ($(\epsilon,\delta)$-DP) if, for every pair of neighboring datasets $X$ and $X'$ and every subset $S\in\mathcal{Y}$,
        \[ \mathbb{P}\left(\mathcal{M}(X)\in S\right) \leq e^{\epsilon}\mathbb{P}\left(\mathcal{M}(X')\in S\right)+\delta, \]
\end{defn}

\noindent where $\epsilon$ is referred as a privacy budget. We will refer to $(\epsilon,0)$-DP as $\epsilon$-DP. Intuitively, smaller values of $\epsilon$ and $\delta$ lead to a stronger privacy guarantee. Conversely, higher values of $\epsilon$ and $\delta$ imply a weaker guarantee with possibly better utility/accuracy of the released data. 
	
A related notion of differential privacy is the concentrated differential privacy, which aims to control the moments of the \emph{privacy loss variable}: $f(Y)=\mathbb{P}(\mathcal{M}(X)=Y)/\mathbb{P}(\mathcal{M}(X')=Y)$, where $Y$ is distributed as $\mathcal{M}(X)$.
\begin{defn}[R\'enyi Divergence]\label{def:renyi}
    Let $P$ and $P'$ be probability densities. The R\'enyi divergence of order $\lambda\in (1,\infty)$ between $P$ and $P'$ is defined as:
    \begin{align*}
        \operatorname{D}_{\lambda}(P\| P')&=\frac{1}{\lambda-1}\log\int P(y)^{\lambda}P'(y)^{1-\lambda} \ dy\\ 
                                             &= \frac{1}{\lambda-1}\log\mathbb{E}_{y\sim P}\left[\frac{P(y)^{\lambda-1}}{P'(y)^{\lambda-1}}\right].
    \end{align*}
\end{defn}
\begin{defn}[Concentrated Differential Privacy~\cite{Bun2016}]\label{def:cdp}
    An algorithm $\mathcal{M}:\mathcal{X}\to\mathcal{Y}$ satisfies $\rho$-\emph{zero-concentrated differential privacy} ($\rho$-zCDP) if, for every pair of neighboring datasets $X$ and $X'$ and every $\lambda\in (1,\infty)$,
    \[ \operatorname{D}_{\lambda}(\mathcal{M}(X)\|\mathcal{M}(X'))\leq \lambda\rho. \]
\end{defn}

	One useful property of the differential privacy is that it is preserved under post-processing.
    \begin{prop}[Post-processing \cite{Dwork2014}]\label{prop:post}
        For any $(\epsilon,\delta)$-DP ($\rho$-zCDP) algorithm $\mathcal{M}:\mathcal{X}\to\mathcal{Y}$ and arbitrary random function $f:\mathcal{Y}\to\mathcal{Z}$, the algorithm $f\circ\mathcal{M}$ is also $(\epsilon,\delta)$-DP ($\rho$-zCDP).
	\end{prop}

	There may be some certain situations in which we want to apply multiple DP algorithms, e.g., releasing continual or time-series data. In this case, the resulting algorithm is also differentially private. However, every new DP algorithm comes with a cost of privacy loss, as stated in the following proposition.
	\begin{prop}[Composition~\cite{Dwork2014}]\label{prop:comp}
        For any $(\epsilon,\delta)$-DP ($\rho$-zCDP) algorithms $\mathcal{A}_i:\mathcal{X}\to\mathcal{Y}_i$ for $i\in [k]$, the algorithm $\mathcal{A}_{[k]}:\mathcal{X}\to \prod_{i=1}^k \mathcal{Y}_k$ defined by $\mathcal{A}_{[k]}(X)=(\mathcal{A}_1(X),\ldots,\mathcal{A}_k(X))$ is $(k\epsilon,k\delta)$-DP ($k\rho$-zCDP).
	\end{prop}

    To introduce one of the most ubiquitous $\epsilon$-DP algorithms, we start with the $\ell_1$-\emph{sensitivity} of a randomized algorithm $\mathcal{M}:\mathcal{X}\to\mathbb{R}^k$, which is the maximum $\ell_1$  change in the output as a result of modifying a single datum. 
We denote this sensitivity as $\Delta^{\mathcal{M}}$, and formally define it as:
\[\Delta^{\mathcal{M}} = \max_{\text{neighbor }X,X'} \|\mathcal{M}(X)-\mathcal{M}(X')\|_1.\]
\begin{theorem}[Laplace mechanism~\cite{Dwork2014}]\label{thm:lm}
    Let $\mathcal{M}:\mathcal{X}\to\mathbb{R}^k$ be an algorithm with sensitivity $\Delta^{\mathcal{M}}$ and $Y_i$ be a noise generated by sampling from a Laplace distribution at scale $= \Delta^{\mathcal{M}}/\epsilon$, i.e., $Y_i\sim\text{Laplace}(\Delta^{\mathcal{M}}/\epsilon)$, then the randomized algorithm $\mathcal{A}$ defined by
    \[\mathcal{A}(X)=\mathcal{M}(X)+(Y_1,\ldots,Y_k)  \]
is $\epsilon$-DP.
\end{theorem}

In addition to the Laplace mechanism, the Gaussian mechanism is also commonly used to provide $\rho$-zCDP:
\begin{theorem}[Gaussian mechanism~\cite{Bun2016}]\label{thm:gauss}
    Let $\mathcal{M}:\mathcal{X}\to\mathbb{R}^k$ be an algorithm with sensitivity $\Delta^{\mathcal{M}}$ and $Y_i$ be a noise generated by sampling from a Gaussian distribution at scale $ \Delta^{\mathcal{M}}/\sqrt{2\rho}$, i.e., $Y_i\sim N(0,(\Delta^{\mathcal{M}})^2/2\rho)$, then the randomized algorithm $\mathcal{A}$ defined by
    \[\mathcal{A}(X)=\mathcal{M}(X)+(Y_1,\ldots,Y_k)  \]
    is $\rho$-zCDP.
\end{theorem}

In view of Proposition~\ref{prop:comp}, a composition of $N$ Laplace mechanisms at scale $N\Delta^{\mathcal{M}}/\epsilon=O(N)$ is $\epsilon$-DP, while that of $N$ Gaussian mechanisms at scale $\Delta^{\mathcal{M}}\sqrt{N/2\rho}=O(\sqrt{N})$ is $\rho$-zCDP. We see that, for successive use of a DP mechanism, the Gaussian mechanism gives comparatively smaller noise than the Laplace mechanism. The following lemma shows how the two definitions of differential privacy are related.
\begin{lemma}[\cite{Bun2016}]\label{lemma:bun}
    Any $\rho$-zCDP algorithm is also an $(\epsilon,\delta)$-DP algorithm for any given $\delta>0$ and
    \begin{equation}\label{lemma:cdp-adp}
        \epsilon = \rho+2\sqrt{\rho\log(1/\delta)}.
    \end{equation}
    Conversely, for any given $\epsilon$ and $\delta>0$, any $\rho$-zCDP algorithm where
    \begin{equation}\label{eq:adp-cdp}
        \rho = \left(\sqrt{\log(1/\delta)+\epsilon}-\sqrt{\log(1/\delta)}\right)^2,
    \end{equation}
    is also an $(\epsilon,\delta)$-DP algorithm.
\end{lemma}

\begin{figure*}[!h]
\centering
\includegraphics[width=.8\linewidth]{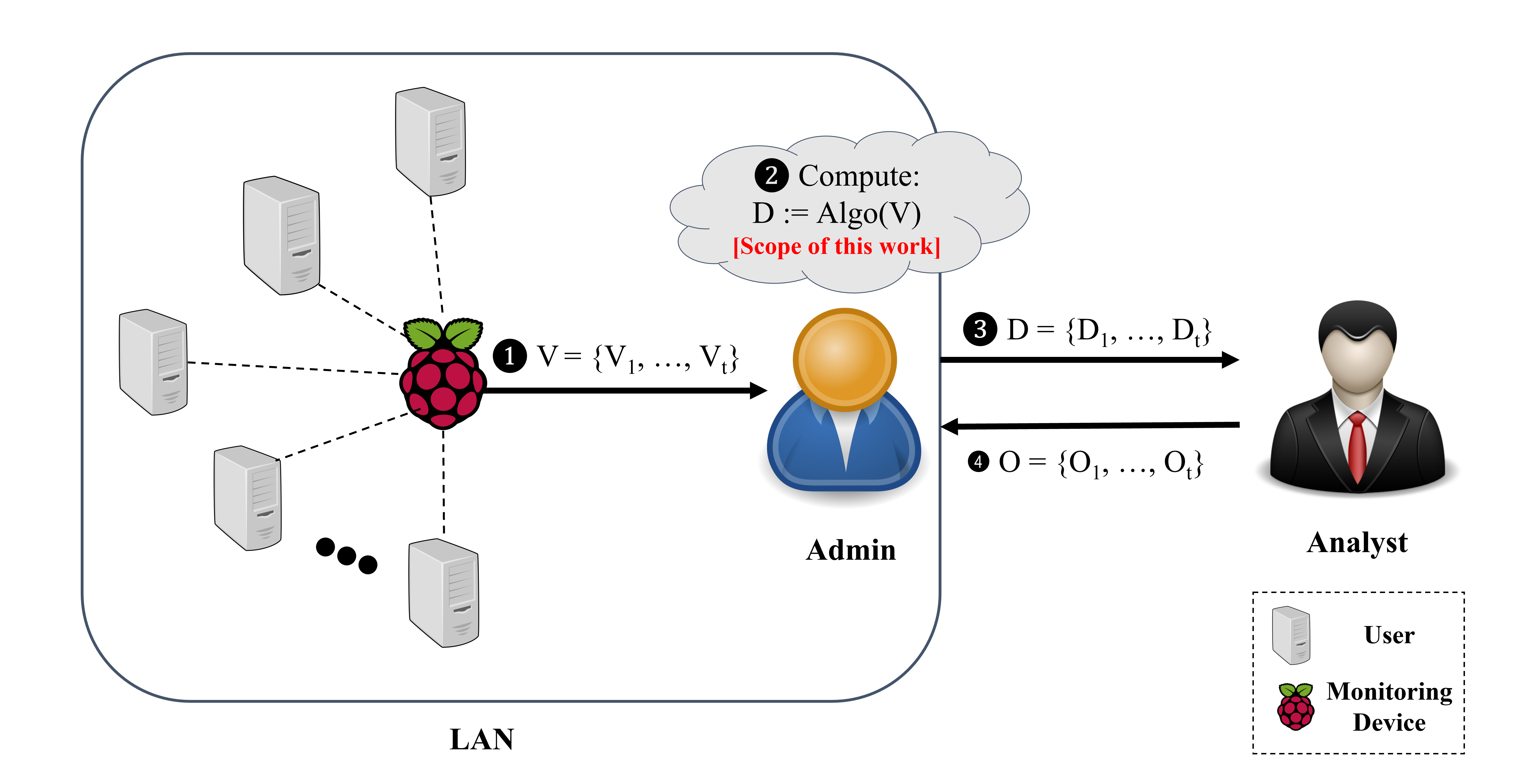}
\caption{System model considered in this work}
\label{fig:scenario}
\end{figure*}

	\section{System and Adversarial Models}\label{sec:prob}
	
	Figure~\ref{fig:scenario} illustrates the system model considered in this work.
	We consider a system in which an entity, called \admin, possesses a LAN consisting of $n$ \user-s (i.e., computing devices).
	In addition, \admin introduces a monitoring device to this LAN in order to observe ARP requests of all \user-s.
	We denote $V_{jk}$ to be aggregate ARP requests originated from \user $k$, measured and accumulated at the $j^{th}$ interval. 
    In this work, we assume the time interval to be in a unit of ``a week", since this time scale allows us to use data collected from a long period of time without losing too much privacy budget from the composition (Proposition~\ref{prop:comp}). $V_j$ is denoted the result after appending all ARP requests of all \user-s generated in week $j$, i.e. $V_j = \{V_{j1}, V_{j2}, ..., V_{jn}\}$.
	
	As shown in Figure~\ref{fig:scenario}, our system starts by having the monitoring node (periodically) send aggregate ARP requests -- $V = \{V_1, ..., V_t\}$ -- to \admin, corresponding to step \ding{182} in Figure~\ref{fig:scenario}. 
	\admin is interested in learning whether the LAN \emph{as a whole} has had any anomalous activities for the last $t$ weeks \textbf{in a private way}.
	Thus, in step \ding{183}, he proceeds to apply a certain algorithm $Algo$ with the goal of hiding sensitive information from the input $V$ and then releases the output $D$ to an external entity \analyst in step \ding{184}. 
	In step \ding{185}, \analyst in turn performs an anomaly detection analysis on $D$ and returns the result $O$ back to \admin. $O$ contains $O_i$ that allows \admin to identify whether the LAN contains an anomaly at week $i$. We summarize notation used throughout the paper in Table~\ref{tab:notation}~\\

	\begin{table}[!htb]
		\caption{Notation}
		\begin{center}
			\resizebox{\columnwidth}{!}{%
			\fbox{\small
				\begin{tabular}{r p{6.3cm} }
					\multicolumn{2}{c}{\underline{Differential Privacy (DP) Notation}}\\\\
					$\epsilon$		& Privacy budget \\
					$\delta$		& Probability of failing DP	guarantees	\\
					$\Delta^{\mathcal{M}}$ & Sensitivity of algorithm $\mathcal{M}$\\
                    $\operatorname{Laplace}(b)$ & Laplace distribution with mean $0$ and scale $b$  \\
					$N(\mu, \sigma^2)$ & Normal distribution with  mean $\mu$ and standard deviation $\sigma$ \\
					\hline \\
					\multicolumn{2}{c}{\underline{System Notation}}\\\\
					$n$ & Number of LAN \user-s \\
					$t$ & Number of data collection intervals \\
					$V_{jk}$ & \user $k$'s ARP requests aggregate at interval $j$ \\
					$V_j = \{V_{j1}, ..., V_{jn}\}$ & Aggregate ARP requests of all LAN \user-s at interval $j$\\
					$V = \{V_1, ..., V_t\}$ & Aggregate ARP requests of all LAN \user-s from interval $1$ to $t$ \\
					$D = \{D_1, ..., D_t\}$ & Output after applying privacy-preserving algorithm \\
					$O = \{O_1, ..., O_t\}$ & Anomaly detection output \\
					
				\end{tabular}
			}
		}
		\end{center}
		\label{tab:notation}
	\end{table}
	
	\noindent{\bf Adversarial Model:} \analyst is assumed to be honest-but-curious, i.e, he always honestly applies an anomaly detection algorithm on any given input data and returns the correct output to \admin. 
	However, during the process, he may attempt to learn sensitive information about \user-s or their relationship, and use it for his own benefits.~\\
	
	\noindent{\bf Goal \& Scope:} 
	In this work, we focus on addressing privacy concerns in the aforementioned system, where data from LAN is exposed to an external party. Hence, we do not consider other LAN settings capable of handling and processing this data locally, e.g., LANs in a large corporate with its own internal anomaly detection tool.
	
	The goal of this work is to design approaches that can be appropriately used as the algorithm $Algo$ in step~\ding{183} of Figure~\ref{fig:scenario}. 
	In other words, our approaches must allow the process of releasing ARP data with some levels of provable privacy guarantees.
	Besides privacy, utility of the privatized/released data for anomaly detection is also important. We must ensure that the privatized value does not change by a significant amount, compared to the non-privatized counterpart; otherwise, it will not be useful in detecting anomalies.
	

	\section{DP Notions for ARP-request data}\label{sec:notions}
	
	\begin{figure*}[!htp]
		\centering
		\includegraphics[width=.6\linewidth]{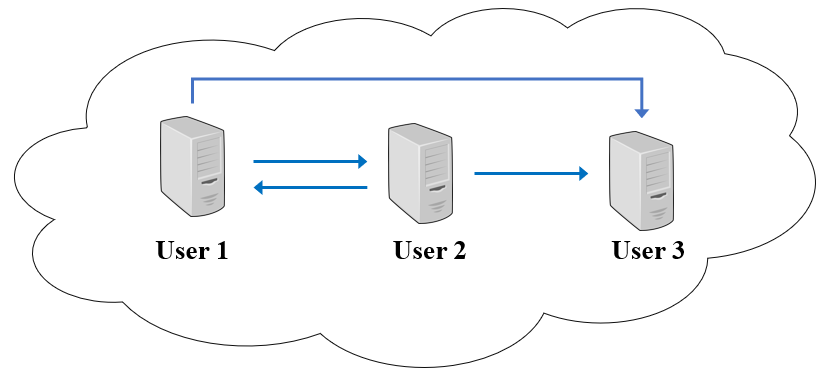}
		\caption{Illustration of a LAN with 3 \user-s and 4 ARP requests (represented by arrows).}
		\label{fig:arp-example}
	\end{figure*}
	
	In this section, we describe 4 variants of differential privacy notions related to our system model.
	The summary of DP notions discussed throughout this Section is shown in Table~\ref{tab:arp-notions}.
	
	\begin{table}[h!]
		\begin{center}
			\caption{Summary of DP notions for ARP-request Data}
			\label{tab:arp-notions}
			\resizebox{\columnwidth}{!}{  
				\begin{tabular}{c|c|c|c}
					Notion & Definition \# & Protected Info. & Protection Prob. \\ & &\\[-1em]
					\hline\hline & &\\[-.8em]
					($\epsilon, \delta$)-edge-DP & \ref{def:edge-delta-dp} & ARP requests & 1-$\delta$ \\ & &\\[-1em]
					$\epsilon$-edge-DP & \ref{def:edge-dp} & ARP requests & 1 \\ & &\\[-1em]
					($\epsilon, \delta$)-node-DP & \ref{def:node-delta-dp} & \user-s & 1-$\delta$  \\ & &\\[-1em]
					$\epsilon$-node-DP & \ref{def:node-dp} & \user-s & 1
				\end{tabular}
			}
		\end{center}
	\end{table}
	
	To understand privacy (i.e., what \emph{concrete} information needs to be private and hidden from \analyst) in our target scenario, we first describe the characteristic of ARP-request data.
	Figure~\ref{fig:arp-example} illustrates an example of a LAN that consists of 3 \user-s producing 4 ARP requests over a specific time interval.
	We define the (ARP-request) ``degree" of \user~$k$ as the number of \user-s that receives ARP requests from \user~$k$. In this example, the degrees of \user $1$, $2$ and $3$ are 2, 2 and 0, respectively.
	
	Using this model, we can view $V_j$ -- aggregate ARP-request data at week $j$ -- as a directed graph, where \user can be represented by a node; whereas
	an arrow (or a directed edge) from node $s$ to node $r$ indicates
	ARP request(s) generated by \user $s$ and sent to \user $r$ in the same time interval. 
	The degree of \user~$k$ is then equivalent to the number of directed edges originating from \user~$k$.
	
	As a directed graph, $V_j$ can not directly represent a database entry, required by Definition~\ref{def:dp}. 
	Thus, the aforementioned notion of differential privacy does not accurately capture the privacy guarantee in our scenario.
	Fortunately, there was prior work focusing on expressing differential privacy of a graph database.
    Specifically, the work in~\cite{Hay2009} presents notions of differential privacy between graphs by first defining two types of neighboring graphs:
    two graphs are \emph{edge-neighboring} if they differ by a single edge. Likewise, they are \emph{node-neighboring} if they differ by a single node.
    
We now proceed to present two notions of privacy in edge-neighboring graphs:

\begin{defn}[($\epsilon, \delta$)-edge-DP]
	\label{def:edge-delta-dp}
	Let $\mathcal{G}$ be the set of graphs between \user-s. An algorithm $\mathcal{M}:\mathcal{G}\to\mathcal{Y}$ satisfies ($\epsilon, \delta$)-edge-differential privacy or ($\epsilon, \delta$)-edge-DP if, for every pair of edge-neighboring graphs $G$ and $G'$ and every subset $S\subseteq\mathcal{Y}$,
	\[ \mathbb{P}\left(\mathcal{M}(G)\in S\right) \leq e^{\epsilon}\mathbb{P}\left(\mathcal{M}(G')\in S\right)+\delta. \]
\end{defn}

\begin{defn}[$\epsilon$-edge-DP]
	\label{def:edge-dp}
	An algorithm satisfies $\epsilon$-edge-differential privacy ($\epsilon$-edge-DP) if and only if it satisfies ($\epsilon, 0$)-edge-DP.
\end{defn}

Since an edge in our system refers to ARP requests between a pair of \user-s, Definition~\ref{def:edge-delta-dp} and~\ref{def:edge-dp} provide privacy protection for these ARP requests.
This means that an algorithm satisfying $\epsilon$-edge-DP/($\epsilon, \delta$)-edge-DP is guaranteed to reveal no information about all ARP requests exchanged between any pair of \user-s,
\emph{resulting in hiding the ARP relationship of all \user-s}.
This, for example, could hide the source of infection in LAN as it is common for malware to utilize ARP as the first step to discover and infect other LAN \user-s. 

Nonetheless, the guarantee provided by these definitions is not strong enough to protect privacy of individual \user-s. To achieve this stronger guarantee, we adopt the following notions:

\begin{defn}[($\epsilon, \delta$)-node-DP]
	\label{def:node-delta-dp}
	Let $\mathcal{G}$ be the set of graphs between \user-s. An algorithm $\mathcal{M}:\mathcal{G}\to\mathcal{Y}$ satisfies ($\epsilon,\delta$)-node-differential privacy or ($\epsilon,\delta$)-node-DP if, for every pair of node-neighboring graphs $G$ and $G'$ and every subset $S\subseteq\mathcal{Y}$,
	\[ \mathbb{P}\left(\mathcal{M}(G)\in S\right) \leq e^{\epsilon}\mathbb{P}\left(\mathcal{M}(G')\in S\right) + \delta. \]
\end{defn}

\begin{defn}[$\epsilon$-node-DP]
	\label{def:node-dp}
	An algorithm satisfies $\epsilon$-node-differential privacy ($\epsilon$-node-DP) if and only if it satisfies ($\epsilon, 0$)-node-DP.
\end{defn}

Indeed, by removing a node we also have to remove all of its edges. One then has that $\epsilon$-node-DP is stronger than $\epsilon$-edge-DP. In our scenario, an algorithm satisfying $\epsilon$-node-DP/($\epsilon,\delta$)-node-DP prevents information leakage about presence or absence of any individual \user.~\\

\noindent\textbf{Remark:}
recall $\delta$ represents an upper bound of the probability that an algorithm fails to satisfy the $\epsilon$-DP notion. 
As an example, an algorithm satisfying $(\epsilon, \delta)$-node-DP has at most $\delta$ probability that will leak some information about an individual node in a graph.
To make $(\epsilon, \delta)$-edge/node-DP notions meaningful in practice, one must minimize this failure probability by ensuring that $\delta$ is negligible in terms of number of data points ($\#p$) considered in the DP notion~\cite{Dwork2014}. One way to achieve this is to set $\delta$ to:

\begin{center}
	$\delta = \delta^\prime/\#p $ for some small $\delta^\prime$
\end{center}

In $(\epsilon, \delta)$-node-DP notion, $\#p$ is the number of nodes;
whereas, in $(\epsilon, \delta)$-edge-DP, $\#p$ corresponds to the number of possible directed edges $\approx (\#nodes)^2$. 
Thus, it is easy to see that $\delta$ in $(\epsilon, \delta)$-edge-DP must be set smaller than that in $(\epsilon, \delta)$-node-DP in order to attain the negligible probability.

	\section{Releasing ARP-request Data with $\epsilon$-edge/node-DP}\label{sec:arp-eps}
	In this section, we present two approaches, called \emph{na\"ive} and \emph{histogram-based}; the former guarantees $\epsilon$-edge-DP while the latter is proven to satisfy the $\epsilon$-node-DP notion. Later in Section~\ref{sec:arp-delta}, we describe variants of these approaches that satisfy the more relaxed $(\epsilon, \delta)$-edge/node-DP notions.
	
	\subsection{Na\"ive Approach}
	
	\begin{algorithm}[!ht]
		\caption{Na\"ive Approach}
		\label{alg:naive}
		\DontPrintSemicolon
		\KwIn{$V=\{V_1, V_2, ..., V_t\}$, $t$, $\epsilon$}
		\KwOut{$D = \{D_1, D_2, ..., D_t\}$}
		
		\For{$j =1$ \KwTo $t$}{
			$D_j \gets \textsc{Sum}(\textsc{Degree}(V_j))$ \;
			$D_j \gets  D_j +  \text{Laplace}(t/\epsilon)$ \; 
			\lIf{$D_j > 0$}{$D_j \gets \text{int}(D_j)$}
			\lElse{$D_j \gets 0$}
		}
	\end{algorithm}
	
	The na\"ive approach is described in Algorithm~\ref{alg:naive}. 
	In the rest of this section, we discuss non-trivial details of this approach and show that it indeed satisfies $\epsilon$-edge DP.

	\begin{theorem}\label{thm:naive}
		The na\"ive approach as described in Algorithm~\ref{alg:naive} is $\epsilon$-edge-DP.
	\end{theorem}
	
	\begin{proof}
    Let $V_{j}\in\mathcal{G}$ be the directed graph of ARP requests in week $j$. Let $\mathcal{M}$ be the algorithm that computes the weekly total degrees and $D_j=\mathcal{M}(V_j)$  (Line~2 of Algorithm~\ref{alg:naive}), which also corresponds to the total number of edges in $V_j$. To preserve $\epsilon$-edge-DP of each \user's ARP requests, one can simply use the Laplace mechanism. To do so, we need to find an upper bound of the sensitivity $\Delta^{\mathcal{M}}$. Let $V'_{j}$ be an edge-neighboring graph of $V_j$ in week $j$ and $D'_j=\mathcal{M}(V'_j)$ . Then, $\Delta^{\mathcal{M}} = |D_j-D'_j|\leq 1$ and we have the following Laplace mechanism $\mathcal{A'}$ (Line~2-3) guarantee the $\epsilon/t$-node DP:
    \[
        \mathcal{A'}(V_{j}) = \mathcal{M}(V_j)+Y_j,
    \] 
    where $Y_{j}\sim \text{Laplace}(t/\epsilon)$ (Line~3). 
    
    Algorithm~\ref{alg:naive} can then be represented as:
    \[
    \mathcal{A}(V) = \mathcal{P}(\mathcal{A'}(V_1), \ldots, \mathcal{A'}(V_t))
    \] 
    where $\mathcal{P}$ is a post-processing function (Line~4-5) that: (i) precludes a negative output by thresholding it to 0, and (ii) rounds a non-negative privatized value into the closest integer in order to prevent the floating point attack~\cite{Ilya2012}.
    
    By Proposition~\ref{prop:post} and~\ref{prop:comp}, we can conclude that this algorithm is $t\epsilon/t$-edge-DP or $\epsilon$-edge-DP.
    
    \end{proof}
    
    To prevent excessive information loss, one needs the Laplace noise to be smaller than $D_j$, i.e., $t/\epsilon < \mathbb{E}[D_j]$ or $\epsilon > t/\mathbb{E}[D_j]$. This can be achieved in realistic settings, e.g., $\epsilon = 2$ in our experiment (Section~\ref{sec:exp}) where $t=30$ and the lower quartile of $D_j$ is $20$.

    On the other hand, a similar analysis for the $\epsilon$-node-DP results in much bigger Laplace noises; consider two node-neighboring directed graphs $V_j,V'_j$ of $n$ \user-s. The degrees $D_j,D'_j$ defined as above satisfy $|D_j-D'_j|\leq n$, which cannot be improved further. Thus, in order to employ the Laplace mechanism, the noises have to be sampled from $\text{Laplace}(tn/\epsilon)$. In contrast to the edge-DP regime, the scale of the noise comes with a factor of $n$. As a result, for a large number of \user-s, it is no longer feasible to preserve both privacy and utility at the same time.
	
    \subsection{Histogram-based Approach}\label{histogram}

	As seen in the previous subsection, the na\"ive approach can not be used  to satisfy $\epsilon$-node-DP in practice due to its high sensitivity, leading to too strong additive noises which in turn significantly lower utility of the released data. Instead, we propose a second approach utilizing a histogram that helps reduce the $\epsilon$-node-DP sensitivity to a reasonable amount.
	
	Our histogram-based approach is shown in Algorithm~\ref{alg:histogram}. 
	The rationale behind this approach is to transform the degree data in such a way that its sensitivity is minimized when any \user is removed from $V_j$.
	Naturally, a histogram is a good fit for this approach since it provides a way to partition data into disjoint groups/bins, where each bin in this case represents a range of degrees.
	Thus, this approach first computes the degrees of each \user in a specific week and uses this degree data to construct a histogram, as shown in Line~2 of Algorithm~\ref{alg:histogram}.
	This histogram data minimizes the $\epsilon$-node-DP sensitivity because removing a \user from the histogram data affects only one bin, i.e., the one this \user belongs, and it only decreases its bin count by one; \emph{other histogram bins are unaffected by this change.}
	We then can apply the Laplace mechanism on each bin (Line~3), threshold and round the resulting value to the closest integer (Line~5-6) and finally return this noisy histogram as an output. 
	
	We now formally show that the histogram-based approach satisfies $\epsilon$-node-DP.

	\begin{theorem}\label{thm:hist}
		The histogram-based approach as described in Algorithm~\ref{alg:histogram} is $\epsilon$-node-DP.
	\end{theorem}
\begin{proof}
   Let $V_j$ and $V'_j$ be node-neighboring directed graph at time $j$, i.e., $V'_j$ can be obtained from $V_j$ by adding or removing a single node. Let $\mathcal{M}:\mathcal{G}\to \mathbb{R}^{k}$ be the algorithm that computes the histogram of the degrees, i.e., the entries of $\mathcal{M}(V_j)$ and $\mathcal{M}(V'_j)$ are the count of nodes by their degrees. Then $\mathcal{M}(V_j)$ and $\mathcal{M}(V'_j)$ differ by one in the entry corresponding to the degree of \user $j$, who only exists in either $V_j$ or $V'_j$. Therefore, $\Delta^{\mathcal{M}} = |\mathcal{M}(V)-\mathcal{M}(V'_j)|\leq 1$. 
   
    Observe that Line~2-7 of Algorithm~\ref{alg:histogram} can be written as a randomized algorithm $\mathcal{A'}:\mathcal{G}\to \mathbb{R}^{k}$ defined by   
    \[
    \mathcal{A'}(V_j) = \mathcal{P}(\mathcal{M}(V_j)+(Y_1,\dots,Y_k))
    \] 
where $Y_i\sim \text{Laplace}(t/\epsilon)$ and $\mathcal{P}$ corresponds to the \emph{threshold-then-round} function computed on all bin counts (Line~5-6). It follows from Theorem~\ref{thm:lm} and Proposition~\ref{prop:post} that $\mathcal{A'}$ is $\epsilon/t$-node-DP.

Then, we can define Algorithm~\ref{alg:histogram} as a randomized algorithm $\mathcal{A}$ as follows:
\[
\mathcal{A}(V) = (\mathcal{A'}(V_1),\ldots,\mathcal{A'}(V_t))
\]  
By Proposition~\ref{prop:comp}, we have that the histogram-based approach (described in Algorithm~\ref{alg:histogram}) is $t\epsilon/t$-node-DP or $\epsilon$-node-DP.

\end{proof}

	\begin{algorithm}[!t]
		\caption{Histogram-based Approach}
		\label{alg:histogram}
		\DontPrintSemicolon
		\KwIn{$V=\{V_1, V_2, ..., V_t\}$, $t$, $\epsilon$}
		\KwOut{$D = \{D_1, D_2, ..., D_t\}$}
		
		\For{$j =1$ \KwTo $t$}{
			$D_j \gets \textsc{Histogram}(\textsc{Degree}(V_j))$ \;
			\ForEach{$bin \in D_j$} {
				$bin.count \gets bin.count + \text{Laplace}(t/\epsilon)$ \;
				\lIf{$bin.count > 0$}{$bin.count \gets \text{int}(bin.count)$}
				\lElse{$bin.count \gets 0$}
			}
		}
	\end{algorithm}

	\section{Releasing ARP-request Data with $(\epsilon, \delta)$-edge/node-DP}\label{sec:arp-delta}
	
    The approaches in the previous section require adding a noise proportional to $t$, which may not scale well in practice when $t$ is large. We explore an alternative by instead adopting the Gaussian Mechanism in order to reduce additive noise from $O(t)$ to $O(\sqrt{t})$. We call these variants, \emph{na\"ive-$\delta$} and \emph{histogram-based-$\delta$}, which guarantee ($\epsilon,\delta$)-edge-DP and ($\epsilon,\delta$)-node-DP respectively.
	
	\subsection{Na\"ive-$\delta$ Approach}
	
    In conjunction with the na\"ive approach (Algorithm~\ref{alg:naive}) which gives a strong privacy guarantee by adding considerably large amount of noises, we develop here another approach that adds less noises, but provides a weaker $(\epsilon,\delta)$-edge DP guarantee. The algorithm is described in Algorithm~\ref{alg:naive-delta}. Similar to Algorithm~\ref{alg:naive}, we round the noisy outputs to the nearest integers to protect the data from floating point attacks. In the rest of this section, we discuss non-trivial details of this approach and show that it indeed satisfies $(\epsilon,\delta)$-edge DP.

	\begin{theorem}\label{thm:naive2}
        The na\"ive-$\delta$ approach as described in Algorithm~\ref{alg:naive-delta} is $(\epsilon,\delta)$-edge-DP.
	\end{theorem}
	
	\begin{proof}
        Let $V_{j}\in\mathcal{G}$ be the directed graph of ARP requests in week $j$. Let $\mathcal{M}$ be the algorithm that computes the weekly total degrees and $D_j=\mathcal{M}(V_j)$  (Line~3 of Algorithm~\ref{alg:naive-delta}). As in the proof of Theorem~\ref{thm:naive}, the edge-sensitivity $\Delta^{\mathcal{M}}$ satisfies $\Delta^{\mathcal{M}} \leq 1$.
           Observe that Line~3-6 of Algorithm~\ref{alg:naive-delta} can be written as a randomized algorithm $\mathcal{A}':\mathcal{G}\to \mathbb{R}^{k}$ defined by   
    \[
    \mathcal{A}'(V_j) = \mathcal{P}(\mathcal{M}(V_j)+(Y_1,\dots,Y_k))
    \] 
where $Y_i\sim N(0,t/2\rho)$ and $\mathcal{P}$ corresponds to the \emph{threshold-then-round} function computed on all bin counts (Line~5-6). It follows from Theorem~\ref{thm:gauss} and Proposition~\ref{prop:post} that $\mathcal{A}'$ is $\rho/t$-zCDP.

Then, we can define Algorithm~\ref{alg:naive-delta} as a randomized algorithm $\mathcal{A}$ as follows:
\[
\mathcal{A}(V) = (\mathcal{A}'(V_1),\ldots,\mathcal{A}'(V_t))
\]  
By Proposition~\ref{prop:comp}, we have that the Algorithm~\ref{alg:naive-delta} is $t\rho/t$-zCDP or $\rho$-zCDP. Using Lemma~\ref{lemma:bun} and recalling the definition of $\rho$ in Line 1 of Algorithm~\ref{alg:naive-delta}, we conclude that this algorithm is also $(\epsilon,\delta)$-edge-DP.

     \end{proof}

	\begin{algorithm}[!ht]
		\caption{Na\"ive-$\delta$ Approach}
		\label{alg:naive-delta}
		\DontPrintSemicolon
		\KwIn{$V=\{V_1, V_2, ..., V_t\}$, $t$, $\epsilon$, $\delta$}
		\KwOut{$D = \{D_1, D_2, ..., D_t\}$}
		
		$\rho \gets \left(\sqrt{\log{(1/\delta)}+\epsilon} - \sqrt{\log{(1/\delta)}}\right)^2$ \;
		\For{$j =1$ \KwTo $t$}{
			$D_j \gets \textsc{Sum}(\textsc{Degree}(V_j))$ \;
			$D_j \gets  D_j +  N(0, t/2\rho)$ \; 
			\lIf{$D_j > 0$}{$D_j \gets \text{int}(D_j)$}
			\lElse{$D_j \gets 0$}
		}
	\end{algorithm}

	\subsection{Histogram-based-$\delta$ Approach}
    We aim to construct an $(\epsilon,\delta)$-node-DP with less noises compared to the $\epsilon$-node-DP algorithm in Section~\ref{histogram}. We still rely on a histogram-based approach as it has small sensitivity upon adding/removing a node. Our histogram-based-$\delta$ approach is described in Algorithm~\ref{alg:histogram-delta}.
	\begin{theorem}\label{thm:hist2}
        The histogram-based-$\delta$ approach as described in Algorithm~\ref{alg:histogram-delta} is $(\epsilon,\delta)$-node-DP.
	\end{theorem}
\begin{proof}
    Let $V_j$ and $V'_j$ be node-neighboring directed graph at time $j$, i.e., $V'_j$ can be obtained from $V_j$ by adding or removing a single node. Let $\mathcal{M}:\mathcal{G}\to \mathbb{R}^{k}$ be the algorithm that computes the histogram of the degrees, i.e., the entries of $\mathcal{M}(V_j)$ and $\mathcal{M}(V'_j)$ are the count of nodes by their degrees. As in the proof of Theorem~\ref{thm:hist}, the node-sensitivity $\Delta^{\mathcal{M}}$ satisfies $\Delta^{\mathcal{M}} \leq 1$
   
    Looking at Algorithm~\ref{alg:histogram-delta}, we observe that Line~3-7 can be written as a randomized algorithm $\mathcal{A}':\mathcal{G}\to \mathbb{R}^{k}$ defined by   
    \[
    \mathcal{A}'(V_j) = \mathcal{P}(\mathcal{M}(V_j)+(Y_1,\dots,Y_k))
    \] 
where $Y_i\sim N(0,t/2\rho)$ and $\mathcal{P}$ corresponds to the \emph{threshold-then-round} function computed on all bin counts (Line~6-7). It follows from Theorem~\ref{thm:gauss} and Proposition~\ref{prop:post} that $\mathcal{A}'$ is $\rho/t$-node-DP.

Then, we can define Algorithm~\ref{alg:histogram-delta} as a randomized algorithm $\mathcal{A}$ as follows:
\[
\mathcal{A}(V) = (\mathcal{A}'(V_1),\ldots,\mathcal{A}'(V_t))
\]  
By Proposition~\ref{prop:comp}, we have that the histogram-based approach (described as in Algorithm~\ref{alg:histogram-delta}) is $t\rho/t$-zCDP or $\rho$-zCDP. From the definition of $\rho$ in Line 1 of Algorithm~\ref{alg:histogram-delta}), we conclude using Lemma~\ref{lemma:bun} that this algorithm is also $(\epsilon,\delta)$-node-DP.
\end{proof}

	\begin{algorithm}[!t]
		\caption{Histogram-based-$\delta$ Approach}
		\label{alg:histogram-delta}
		\DontPrintSemicolon
		\KwIn{$V=\{V_1, V_2, ..., V_t\}$, $t$, $\epsilon$, $\delta$}
		\KwOut{$D = \{D_1, D_2, ..., D_t\}$}
		
		$\rho \gets \left(\sqrt{\log{(1/\delta)}+\epsilon} - \sqrt{\log{(1/\delta)}}\right)^2$ \;
		\For{$j =1$ \KwTo $t$}{
			$D_j \gets \textsc{Histogram}(\textsc{Degree}(V_j))$ \;
			\ForEach{$bin \in D_j$} {
				$bin.count \gets bin.count + N(0, t/2\rho)$\;
				\lIf{$bin.count > 0$}{$bin.count \gets \text{int}(bin.count)$}
				\lElse{$bin.count \gets 0$}
			}
		}
	\end{algorithm}

	\section{Evaluation}\label{sec:exp}
	
	In this section, we evaluate our approaches by deploying them as part of a large-scale research project and reporting their utility from a real-world dataset extracted from such project.
	
	\subsection{Real-world Deployment}\vspace{2mm}
	
	\noindent\textbf{Background.} ASEAN-Wide Cyber-Security Research Testbed Project is a large-scale research project with collaboration between multiple universities primarily located in Southeast Asia including Prince of Songkla University, Thailand (PSU), Universitas Brawijaya, Indonesia (UB), University of Computer Studies Yangon, Myanmar (UCSY), Institute of Technology of Cambodia, Cambodia (ITC), University of Information Technology, Myanmar (UIT), and The University of Tokyo, Japan (UT). The ultimate goal of this project is to create a real-world public testbed of malware behaviors captured in ASEAN countries.

	\begin{figure}[!h]
		\centering
		\includegraphics[width=.8\linewidth]{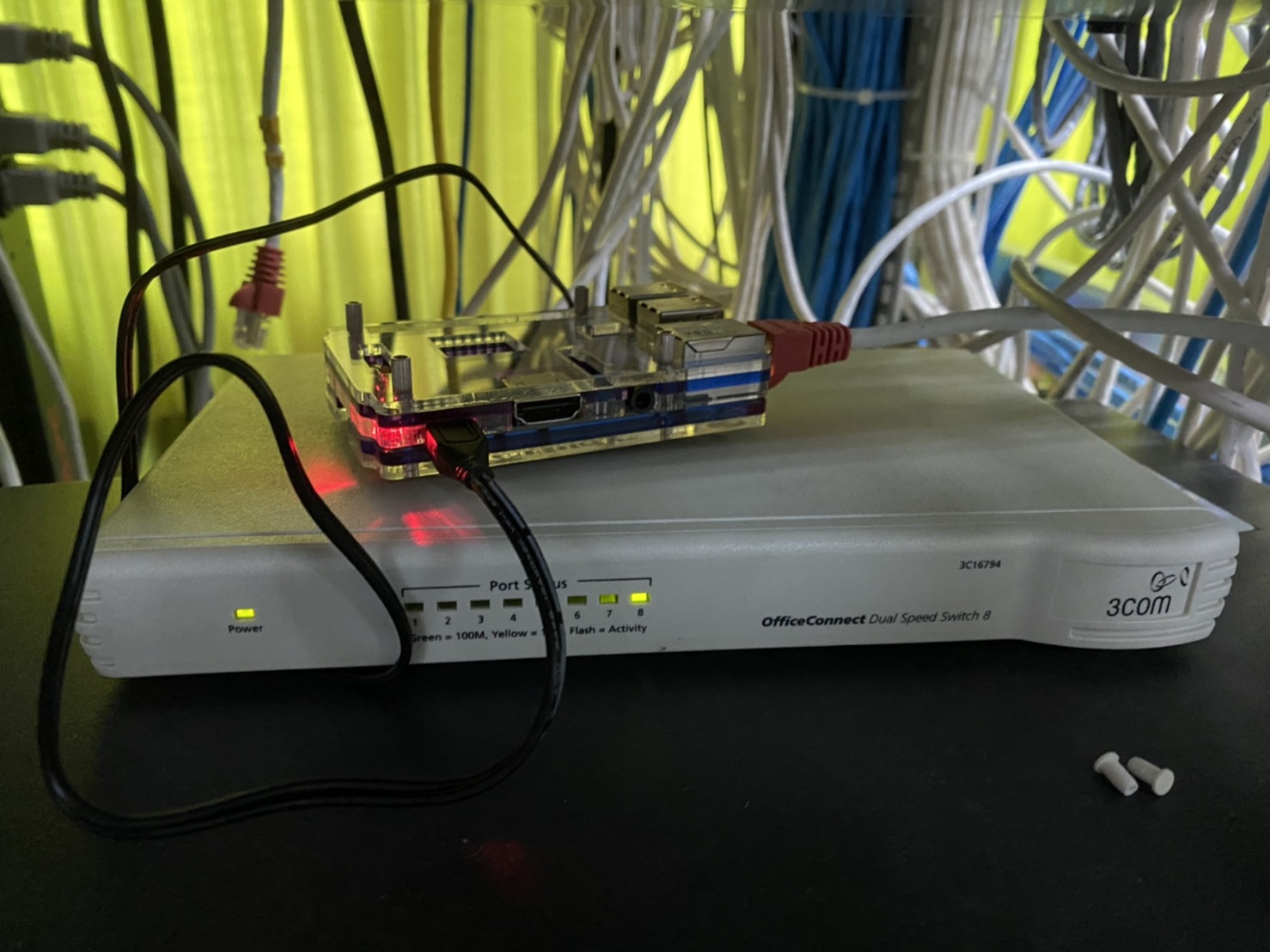}
		\caption{Monitoring device (raspberry-Pi 3B) deployed to a LAN}
		\label{fig:monitor-device}
	\end{figure}
	
	Independent of our work, the first phase of this project involves capturing, collecting and analyzing LAN data in Southeast Asian countries. 
	To achieve this task, a small monitoring device, implemented atop of a raspberry-Pi 3B in Figure~\ref{fig:monitor-device}, is introduced and placed  
	into several LANs across the ASEAN region. 
	This monitoring device observes and captures the network traffic flowing within a LAN and periodically outputs the captured data to our server, in which such data is analyzed and a model of ASEAN malware is eventually created.~\\
	
	\noindent\textbf{Deployment.} 
	Our work plays an important role in the second phase of this research project. It allows us to privately share aggregate ARP data collected from the previous phase with other project members as well as to the public domain.
	Our approaches enable a release mechanism of ARP-request data that still retains the utility of LAN anomaly detection.
	To assess utility, we evaluated our approaches on a subset of data captured and extracted from this research project.
	
	The extracted dataset contains all ARP-request data observed and collected from 3 real-world LANs over a 30-week period. 
	These LANs are located in: (1) The University of Tokyo, Japan (thus, its dataset is labeled as \texttt{JPN}), (2) Prince of Songkla University -- Phuket Campus, Thailand (\texttt{HKT}) and (3) Prince of Songkla University -- Hatyai Campus, Thailand (\texttt{HDY}).
	Details about these monitored LANs can be found in Table~\ref{tab:exp}.~\\ 
	
	\begin{table*}[h!]
		\begin{center}
			\caption{Details of monitored LANs}
			\label{tab:exp}
			\resizebox{.9\textwidth}{!}{  
			\begin{tabular}{c|c|c|c|c|c|c|c}
				
				\multirow{2}{*}{Label} & \multicolumn{3}{c|}{Location of LAN} & \multicolumn{3}{c|}{Collection Period} & \multirow{2}{*}{\shortstack{\# \user-s ($n$)}} \\
				 & University & City & Country & Start Date & End Date & \# Weeks ($t$) & \\
				\hline\hline
				\texttt{JPN} & UT & Tokyo & Japan & Aug 9, 2019 & Mar 6, 2020 & 30 & 95 \\
				\texttt{HKT} & PSU & Phuket & Thailand & Nov 6, 2020 & June 4, 2021 & 30 & 63 \\
				\texttt{HDY} & PSU & Hat Yai & Thailand & Oct 21, 2020 & May 19, 2021 & 30 & 206 \\
			\end{tabular}
			}
		\end{center}
	\end{table*}
	
	\noindent\textbf{Parameter Selection.}
	As we collected ARP requests over a 30-week period, $t = 30$.
	The na\"ive approach involves no other parameters.
	Meanwhile, the histogram-based approach consists of an additional set of parameters: the number of bins and the width of each bin.
	Intuitively, a larger number of bins leads to smaller bin counts. 
	In such case, the noise injected by our approach would become too large, severely decreasing utility of the released data. To avoid this problem, we select the number of histogram bins to be relatively small -- 3. Specifically, we choose
	the first two bins to correspond to the number of \user-s whose
	degrees are 1 and 2, respectively; the third bin contains the
	number of \user-s with degree $\ge 3$.
	
	Finally, the approaches in Section~\ref{sec:arp-delta} consist of another parameter $\delta$.
    Recall from the remark in Section~\ref{sec:notions} that $\delta$ must be negligible with respect to the number of data points ($\#p$). In other words:

	\begin{center}	
	$\delta = \delta^\prime/\#p$ for some small $\delta^\prime$
	\end{center}
	
	In our target system, $\#p$ corresponds to $n$ and $n^2$ for the node-DP and edge-DP notions, respectively; See Table~\ref{tab:exp} for the number of \user-s ($n$) in each monitored LAN. Unless stated otherwise, we use $\delta^\prime = 0.01$ for all experiments. 
	Nonetheless, the impact of different $\delta^\prime$ values on the utility is also assessed in the next subsection.
	
	\subsection{Utility Assessment: RMSE}\label{sec:rmse}~\vspace{2mm}

	\noindent\textbf{RMSE.}
	In the context of differential privacy, one common utility metric is defined as an error between the released privatized values $z^*$ and the non-privatized aggregates $z$.
	We adopt a similar approach and select the root-mean-square error (RMSE) as our first evaluation metric:
	
	\[RMSE = \sqrt{ \frac{1}{n} \sum_{i=1}^{n} \left( z^*[i]-z[i] \right) ^2}\]
	
	\noindent where $z[i]$ and $z^*[i]$ represent the $i^{th}$ data point in $z$ and $z^*$, respectively.
	For the na\"ive approach and its variant, $z[i]$ corresponds to the sum of all \user's ARP degrees observed in week $i$, while $z^*[i]$ refers to the privatized output on the same ARP data. On the other hand, $z[i]$ represents a histogram bin in the histogram-based and histogram-based-$\delta$ approaches.~\\
	
	\noindent\textbf{Impact of $\epsilon$.}
	Recall that $\epsilon$ refers to a privacy budget in the DP notion and
	a lower value of $\epsilon$ implies stronger privacy, while possibly sacrificing utility.
	
	Figure~\ref{fig:various-eps} shows the impact of $\epsilon$ on the utility of the proposed approaches. Unsurprisingly, we achieve lower errors and thus better utility from a higher $\epsilon$. For all 3 monitored LANs, $\epsilon = 5$ seems to be a pragmatic choice in order to maintain a low error ($<10$) for all approaches. 
	
	Next, we show how much utility can be improved by using the approaches in Section~\ref{sec:arp-delta} instead of their counterparts in Section~\ref{sec:arp-eps}. The result, illustrated in Figure~\ref{fig:util-improv1}, suggests that both na\"ive-$\delta$ and histogram-based-$\delta$ approaches enjoy higher utility (i.e., a utility gain) when $\epsilon \le 4$. However, as the $\epsilon$ gets larger, this utility gain becomes smaller;
	in fact, the na\"ive-$\delta$ approach incurs a utility loss when $\epsilon \ge 8$ for all monitored LANs. This result suggests using the approaches in Section~\ref{sec:arp-delta} only when one needs stronger privacy, i.e., small $\epsilon$.
	
	Figure~\ref{fig:util-improv1} also indicates the histogram-based-$\delta$ approach significantly outperforms the na\"ive-$\delta$ approach in terms of the utility gain. For $\epsilon \le 4$, the histogram-based-$\delta$ approach provides $\ge 28\%$ utility gain, while a smaller amount of utility gain ($\le 20\%$) can be realized in the na\"ive-$\delta$ approach. This is expected because the histogram-based-$\delta$ approach introduces a smaller value of $\delta$ (see the remark in Section~\ref{sec:notions}), making the additive noise smaller and thus resulting in the higher utility gain. 
	
	In addition, $n$ also has a direct impact to $\delta$ and hence to the overall utility. As seen in Figure~\ref{fig:util-improv1}, among all monitored LANs, \texttt{HDY} has the highest number of \user-s and therefore suffers the lowest utility gain.~\\
	
	\noindent\textbf{Impact of $\delta^\prime$.} We now assess the impact of $\delta^\prime$ on the utility of our approaches. Figure~\ref{fig:various-delta} shows RMSE of the na\"ive-$\delta$ and histgoram-based-$\delta$ approaches for different values of $\delta^\prime$. As expected, increasing $\delta^\prime$ results in a decrease in RMSE and thus improves the utility of our approaches.
	This decrease is logarithmic as a function of $\delta^\prime$.
	
    The utility gain of the na\"ive-$\delta$ and histgoram-based-$\delta$ approaches with respect to their original counterparts is illustrated in Figure~\ref{fig:util-improv2}. Our approaches benefit from the higher utility gain when $\delta^\prime$ is larger. For most $\delta^\prime$ values, the histogram-based-$\delta$ approach provides a positive utility gain over the histogram-based approach. Meanwhile, a utility gain can be achieved from the na\"ive-$\delta$ approach when $\delta^\prime \ge 10^{-3}$.
	
	This experimental result suggests that both na\"ive-$\delta$ and histogram-based-$\delta$ approaches still provide a utility advantage over their original counterparts even for $\delta^\prime$ smaller than $10^{-2}$ (up to $10^{-3}$ for the na\"ive-$\delta$ approach and $10^{-6}$ for the histogram-based-$\delta$ approach). In practice, one may choose to opt for smaller $\delta^\prime$ if a stronger privacy guarantee is needed.
	
		\begin{figure*}[htp]
		\centering
		\begin{subfigure}[t]{0.33\textwidth}
			\centering
			\includegraphics[width=\linewidth]{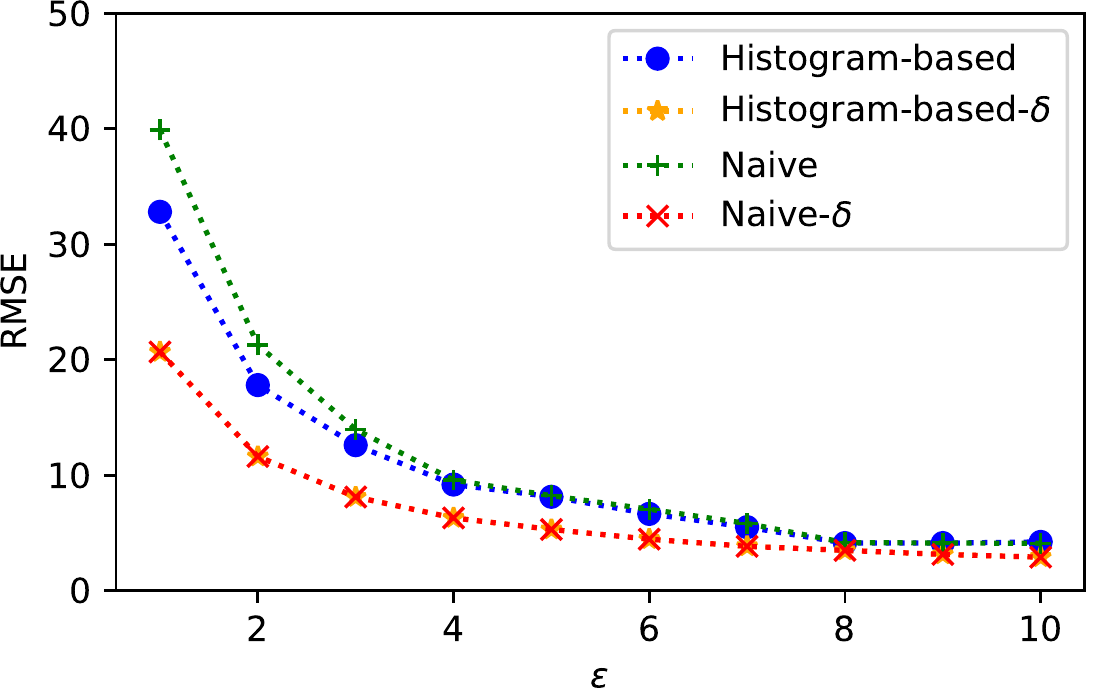}
			\caption{\texttt{JPN}}
		\end{subfigure}%
		~
		\begin{subfigure}[t]{0.33\textwidth}
			\centering
			\includegraphics[width=\linewidth]{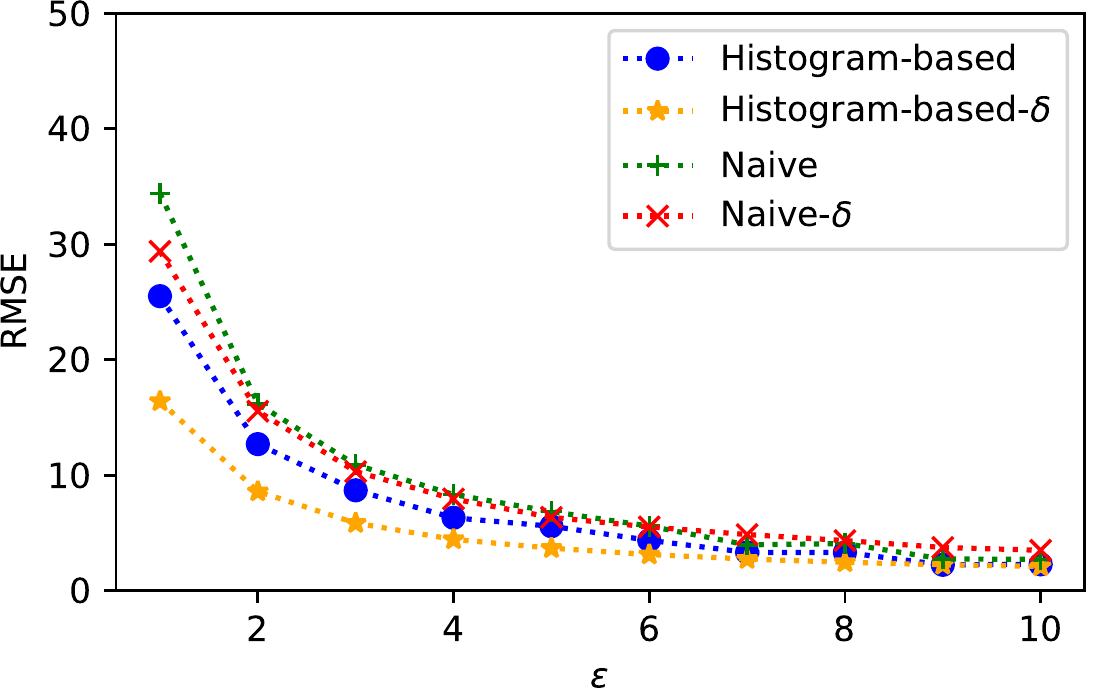}
			\caption{\texttt{HKT}}
		\end{subfigure}%
		~
		\begin{subfigure}[t]{0.33\textwidth}
			\centering
			\includegraphics[width=\linewidth]{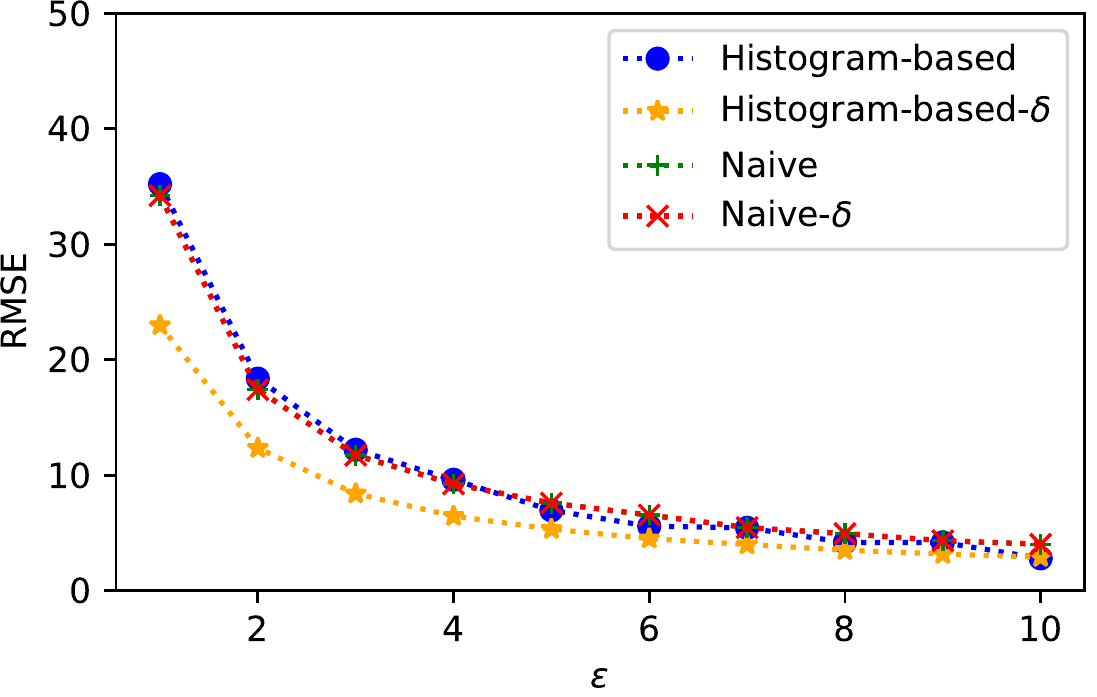}
			\caption{\texttt{HDY}}
		\end{subfigure}
		\caption{RMSE with different $\epsilon$ values}
		\label{fig:various-eps}
		\vspace{3mm}
		
		\centering
		\begin{subfigure}[t]{0.33\textwidth}
			\centering
			\includegraphics[width=\linewidth]{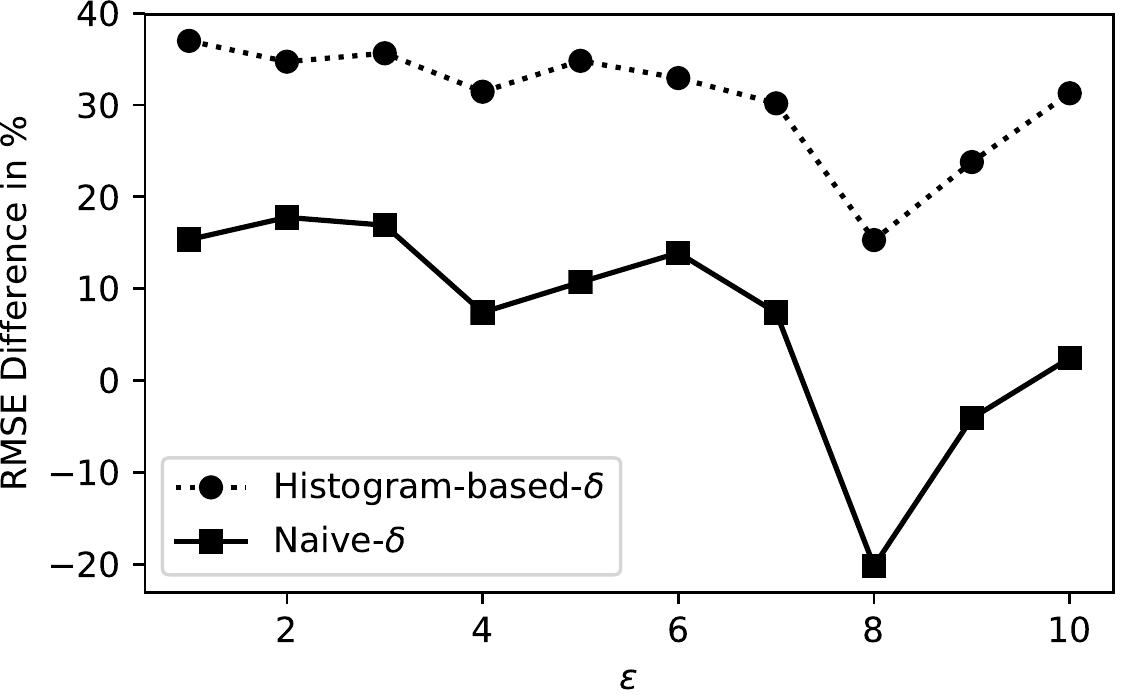}
			\caption{\texttt{JPN}}
		\end{subfigure}%
		~
		\begin{subfigure}[t]{0.33\textwidth}
			\centering
			\includegraphics[width=\linewidth]{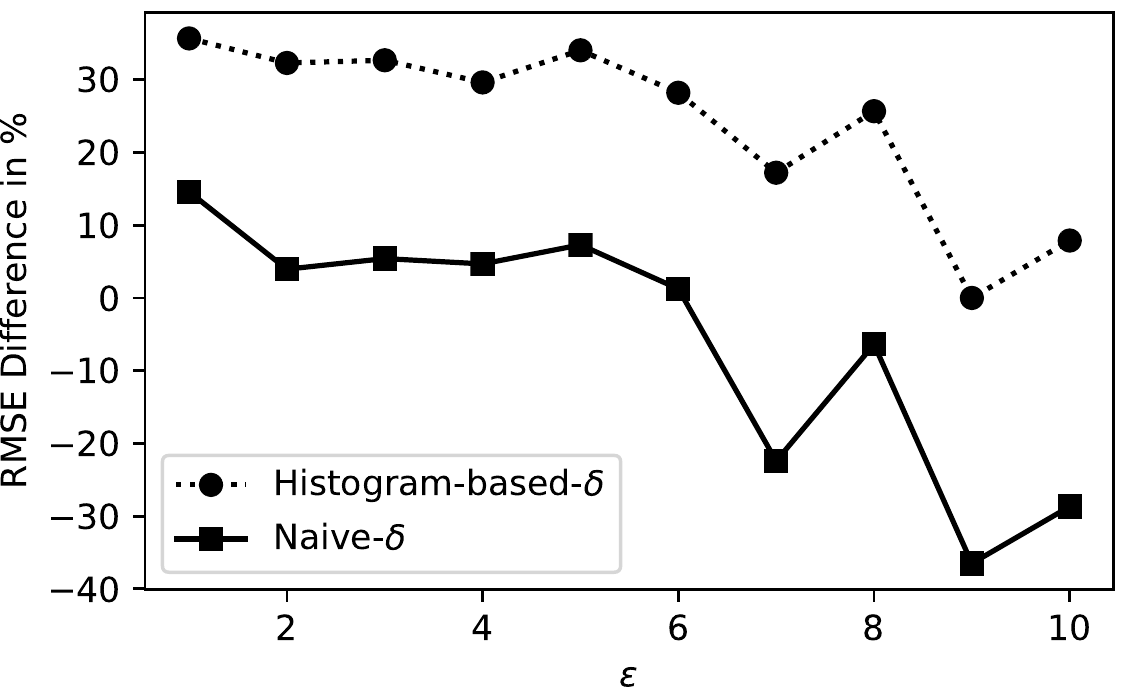}
			\caption{\texttt{HKT}}
		\end{subfigure}%
		~
		\begin{subfigure}[t]{0.33\textwidth}
			\centering
			\includegraphics[width=\linewidth]{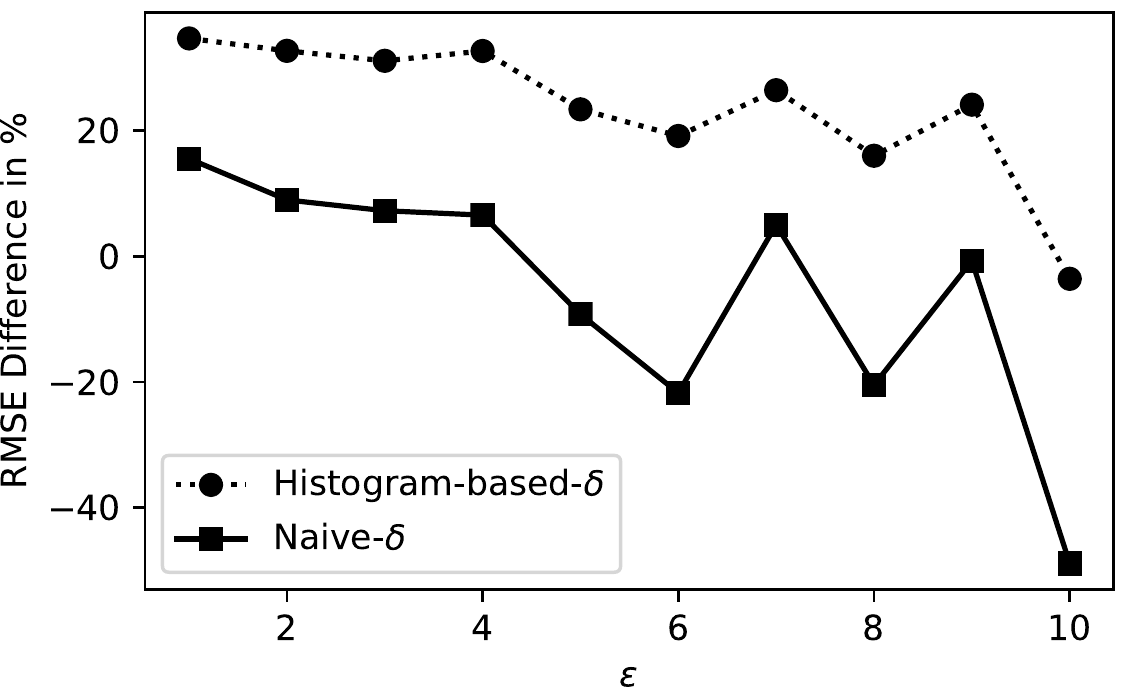}
			\caption{\texttt{HDY}}
		\end{subfigure}
		\caption{Utility gain (in $\%$) with respect to their $\epsilon$-edge/node-DP counterparts}
		\label{fig:util-improv1}
		\vspace{3mm}
		
		\centering
		\begin{subfigure}[t]{0.33\textwidth}
			\centering
			\includegraphics[width=\linewidth]{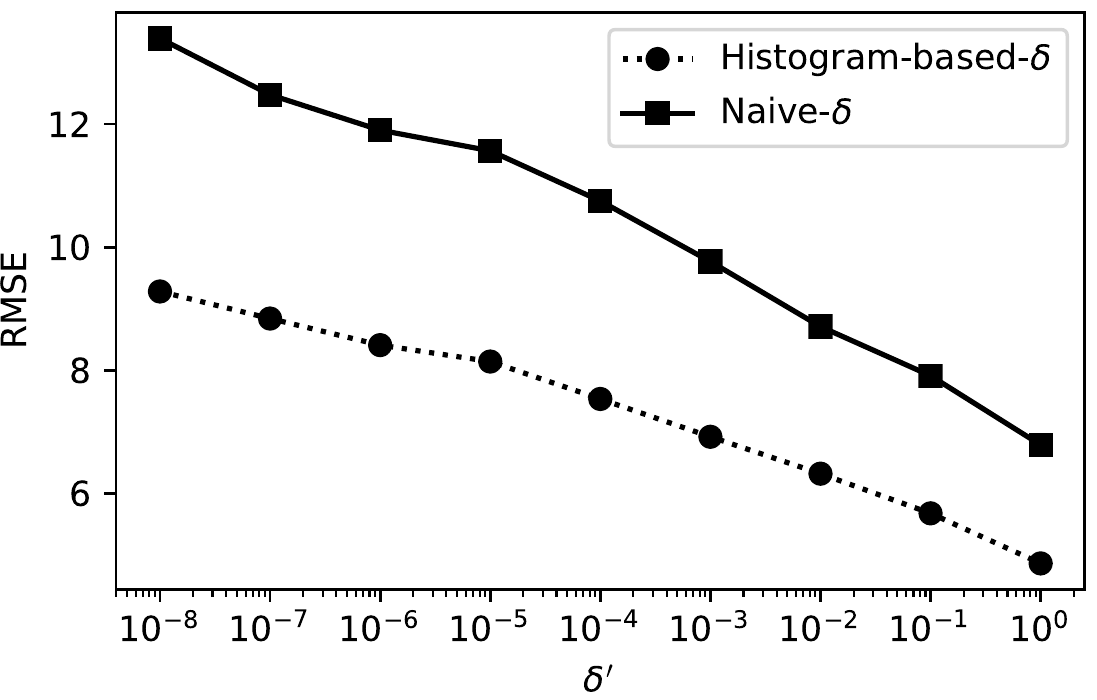}
			\caption{\texttt{JPN}}
		\end{subfigure}%
		~
		\begin{subfigure}[t]{0.33\textwidth}
			\centering
			\includegraphics[width=\linewidth]{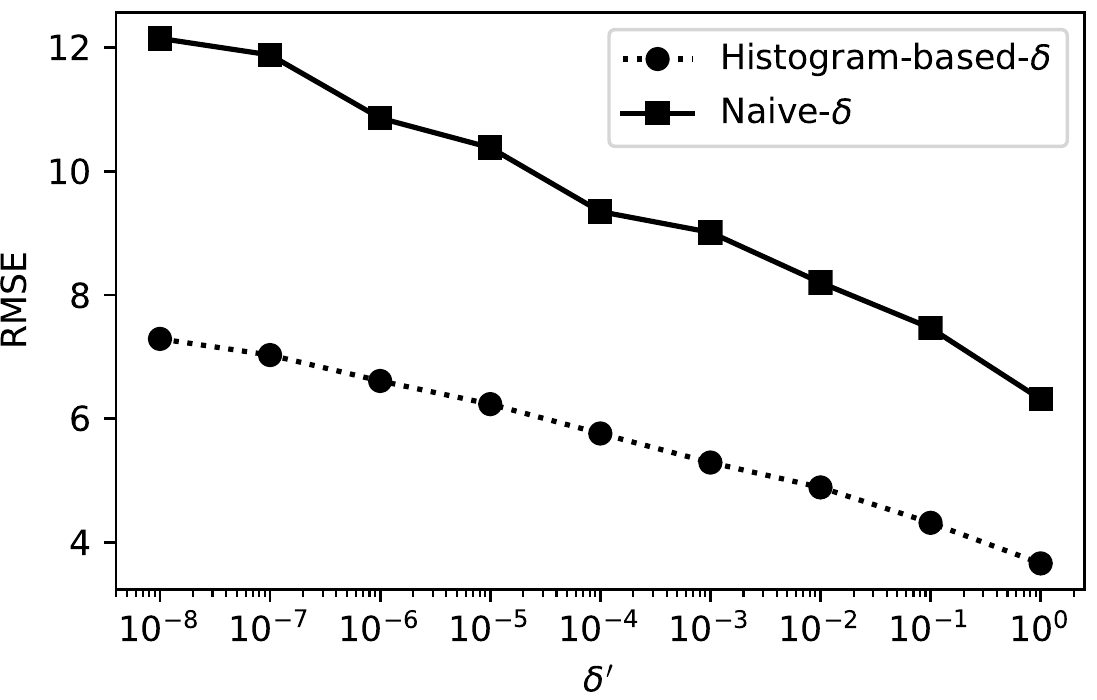}
			\caption{\texttt{HKT}}
		\end{subfigure}%
		~
		\begin{subfigure}[t]{0.33\textwidth}
			\centering
			\includegraphics[width=\linewidth]{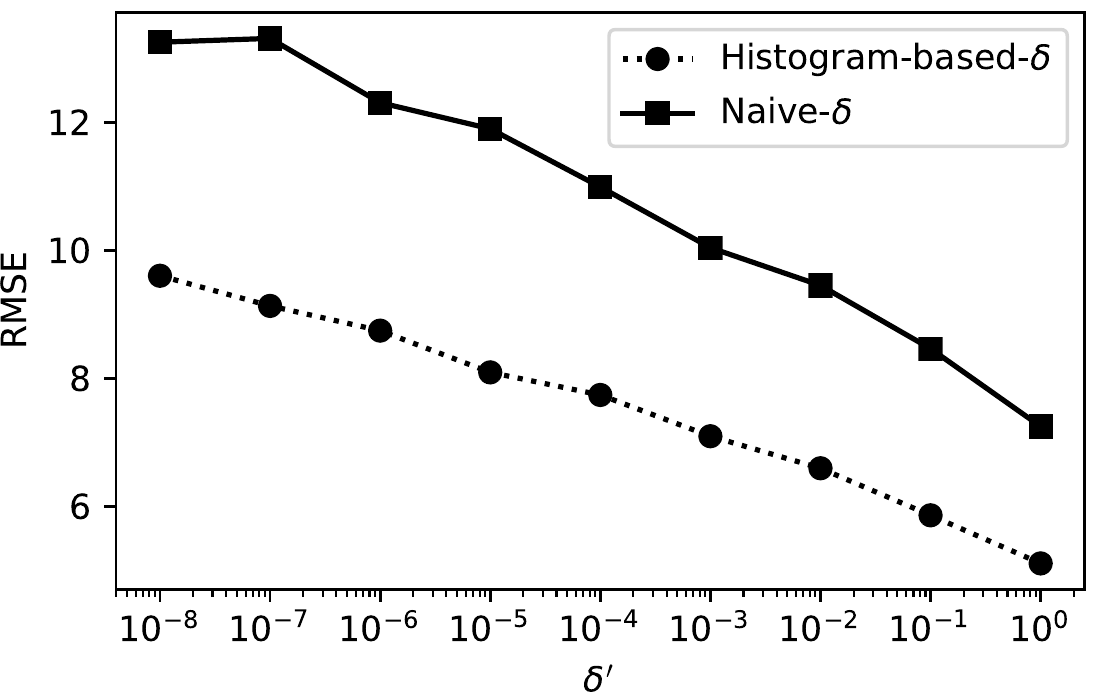}
			\caption{\texttt{HDY}}
		\end{subfigure}
		\caption{RMSE with different $\delta^\prime$ values where $\delta = \delta^\prime/(\#p)$ and $\epsilon$ is fixed to 1}
		\label{fig:various-delta}
		\vspace{3mm}
		
		\centering
		\begin{subfigure}[t]{0.33\textwidth}
			\centering
			\includegraphics[width=\linewidth]{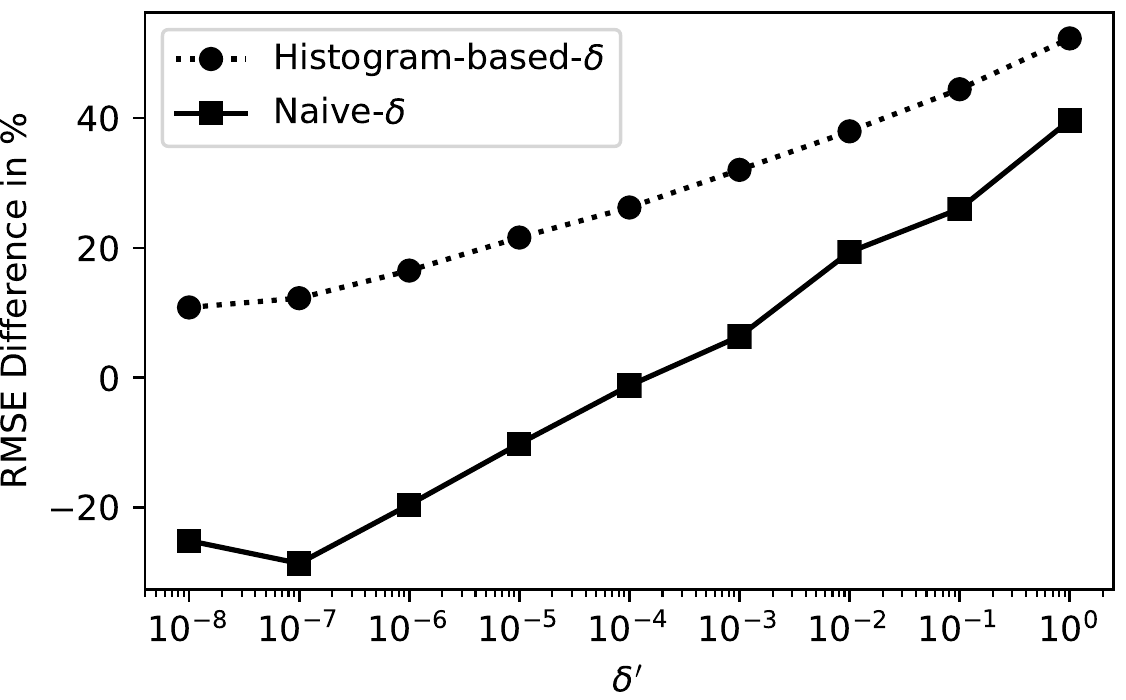}
			\caption{\texttt{JPN}}
		\end{subfigure}%
		~
		\begin{subfigure}[t]{0.33\textwidth}
			\centering
			\includegraphics[width=\linewidth]{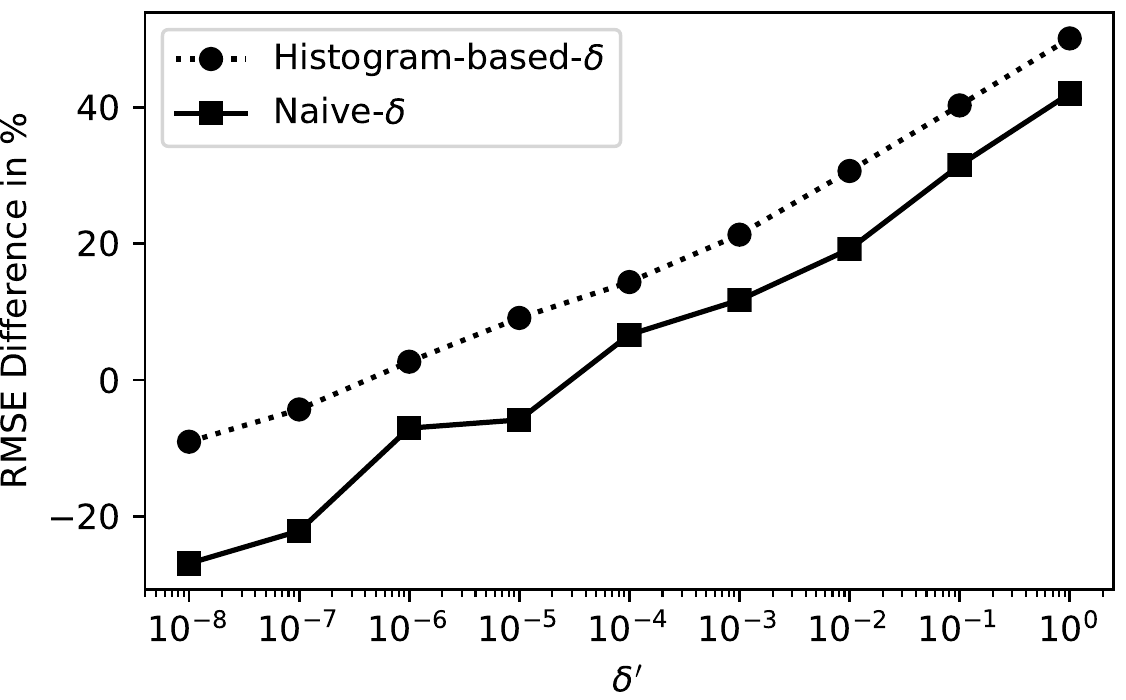}
			\caption{\texttt{HKT}}
		\end{subfigure}%
		~
		\begin{subfigure}[t]{0.33\textwidth}
			\centering
			\includegraphics[width=\linewidth]{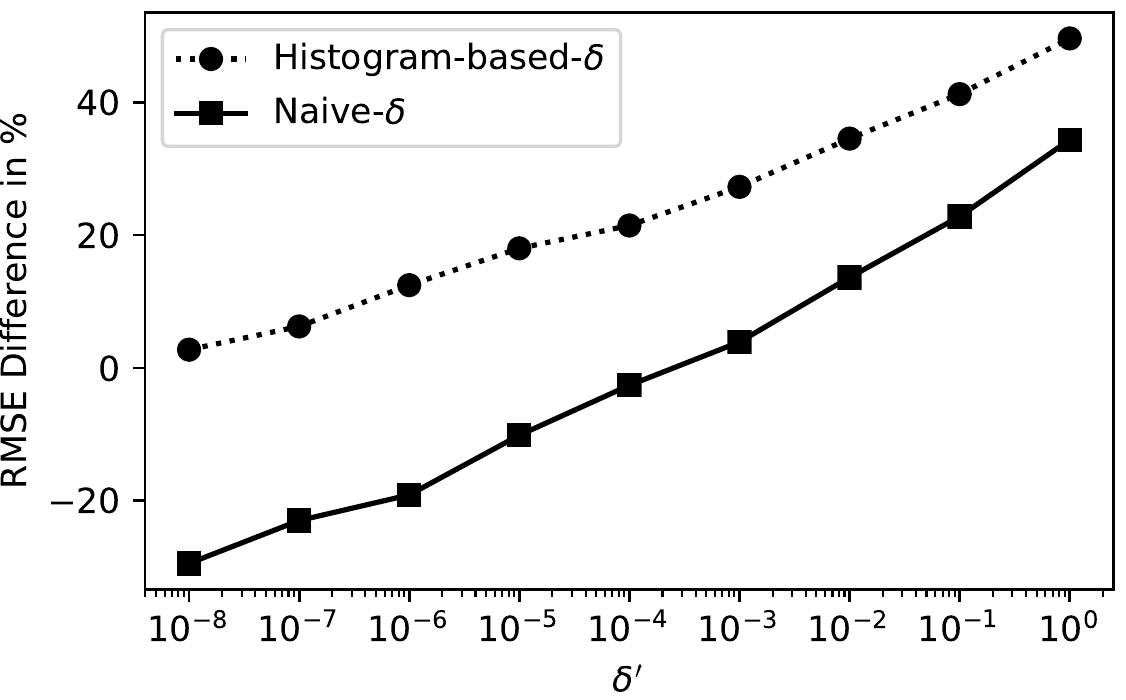}
			\caption{\texttt{HDY}}
		\end{subfigure}
		\caption{Utility gain (in $\%$) with respect to their $\epsilon$-edge/node-DP counterparts}
		\label{fig:util-improv2}
	\end{figure*}
	
	\subsection{Utility Assessment: Anomaly Detection Accuracy}~\vspace{1mm}
	
	\noindent\textbf{Anomaly detection algorithm.} In addition to low errors, it is also essential that outputs produced by our approaches can still be useful in identifying anomalous activities in LAN. 
	Hence, we further evaluate utility of our approaches by assessing them via a LAN anomaly detector. In this experiment, we consider our approaches to preserve the utility of anomaly detection if the anomaly detector classifies the privatized data the same way as the original (non-privatized) data.
	
    For the anomaly detector, we choose an approach based on exponentially weighted moving average and variance~\cite{Montgomery1991} proposed by Matsufuji et.al.~\cite{matsufuji2019arp} since it is tailored specifically for detecting LAN anomalies based on ARP data, which is also the focus in this work. All parameter values are selected based on the recommendation from~\cite{matsufuji2019arp}.
	
	It is worth noting that the anomaly detector in~\cite{matsufuji2019arp} only supports input of type univariate time series. However, the histogram-based approach and its variant produce a multivariate time series output (i.e., a time series of histograms), and hence cannot be used directly as input to the anomaly detector. To address this issue, we perform a simple transformation that converts two consecutive histograms into a single variable using the $L_1$ distance function; the result of this transformation is then given as input to the anomaly detector. More formally, the transformation is defined as:
	\begin{center}
		\[z^*[i] = \|hist[i] - hist[i+1]\|_1 \text{ for } i \in \{1,...,t-1\}\]
	\end{center}
	
	\begin{figure*}[!h]
		\centering
		\begin{subfigure}[t]{0.33\textwidth}
			\centering
			\includegraphics[width=\linewidth]{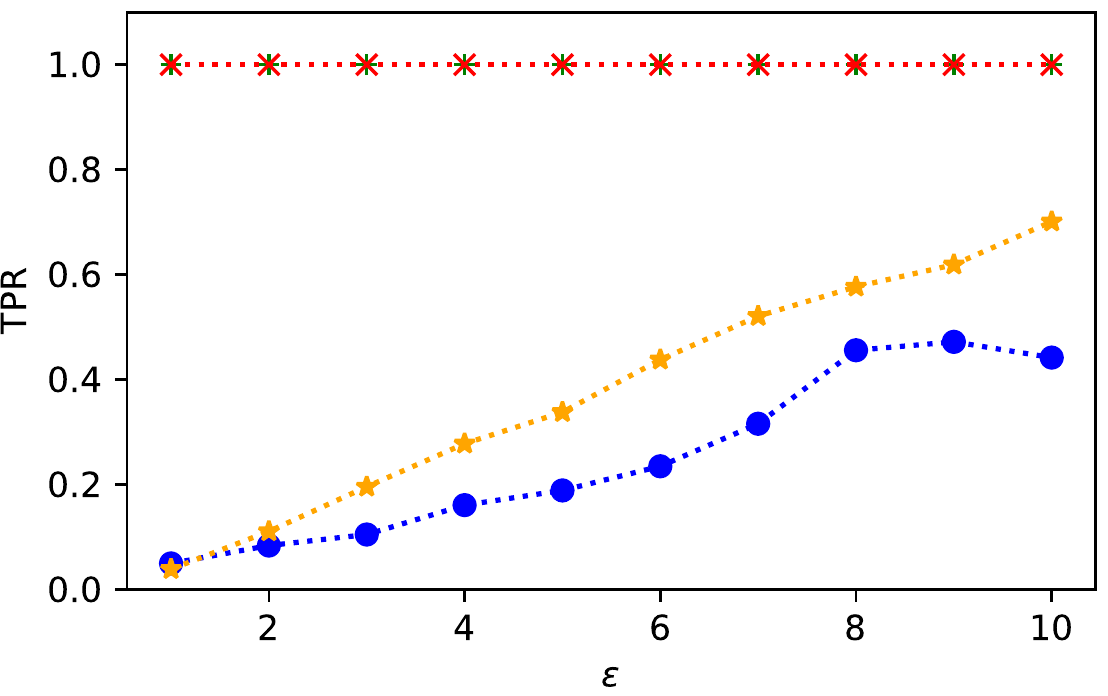}
			\caption{\texttt{JPN}}
		\end{subfigure}%
		~
		\begin{subfigure}[t]{0.33\textwidth}
			\centering
			\includegraphics[width=\linewidth]{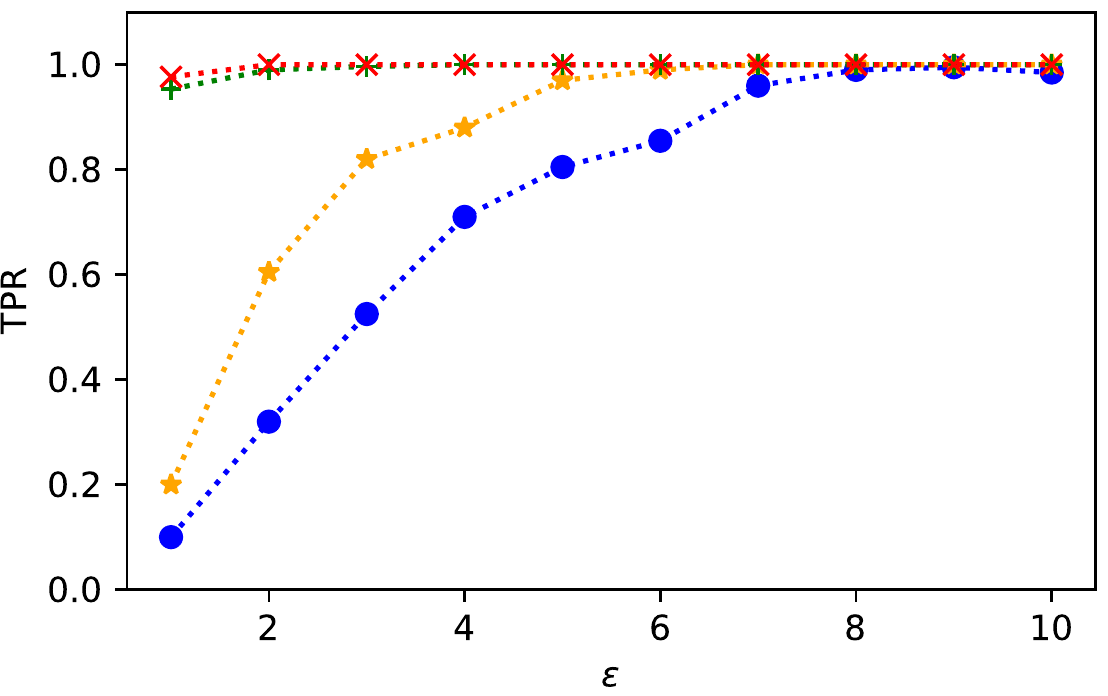}
			\caption{\texttt{HKT}}
		\end{subfigure}%
		~
		\begin{subfigure}[t]{0.33\textwidth}
			\centering
			\includegraphics[width=\linewidth]{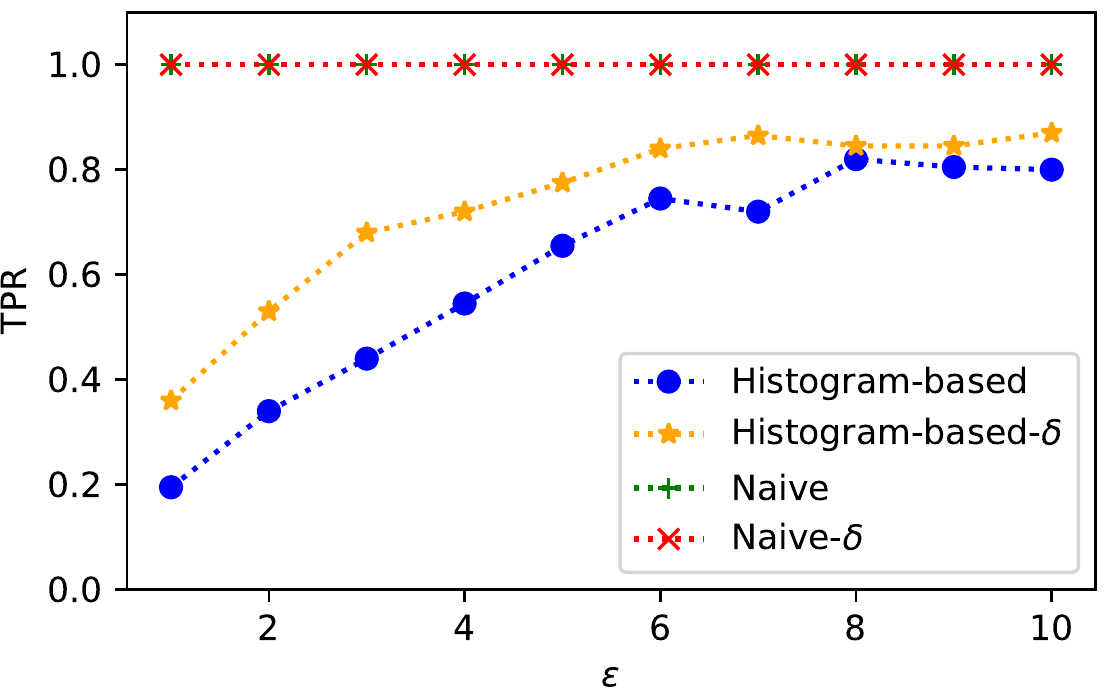}
			\caption{\texttt{HDY}}
		\end{subfigure}%
		\caption{$TPR$ result for different $\epsilon$ in all 3 monitored LANs}
		\label{fig:tpr}
		\vspace{5mm}
		
		\centering
		\begin{subfigure}[t]{0.33\textwidth}
			\centering
			\includegraphics[width=\linewidth]{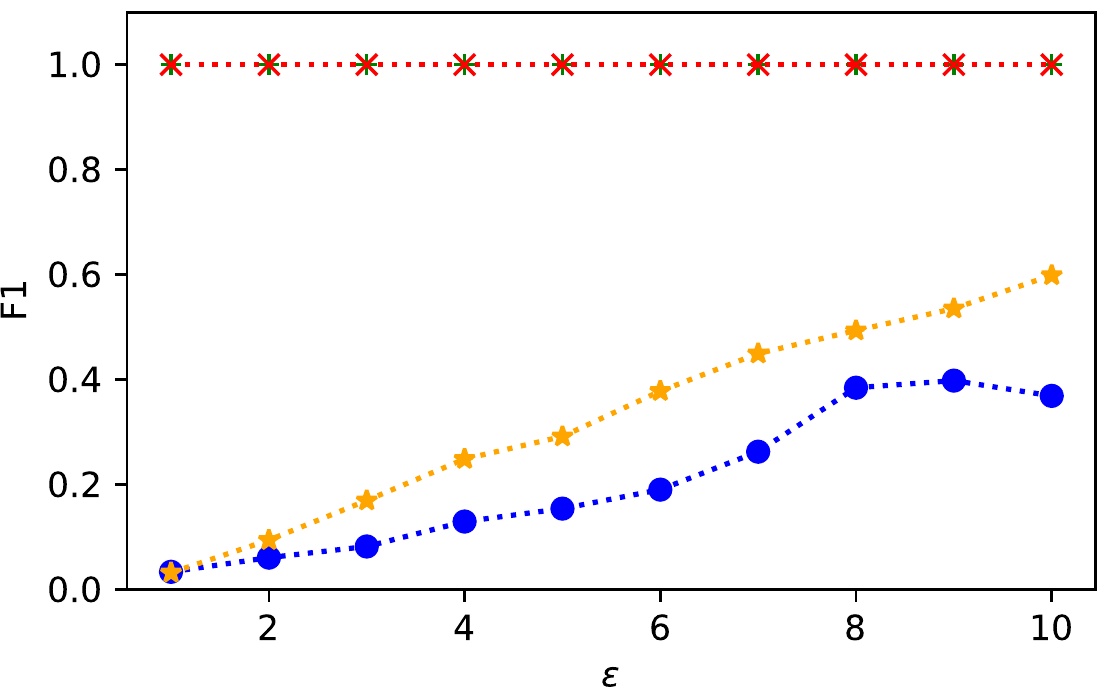}
			\caption{\texttt{JPN}}
		\end{subfigure}%
		~
		\begin{subfigure}[t]{0.33\textwidth}
			\centering
			\includegraphics[width=\linewidth]{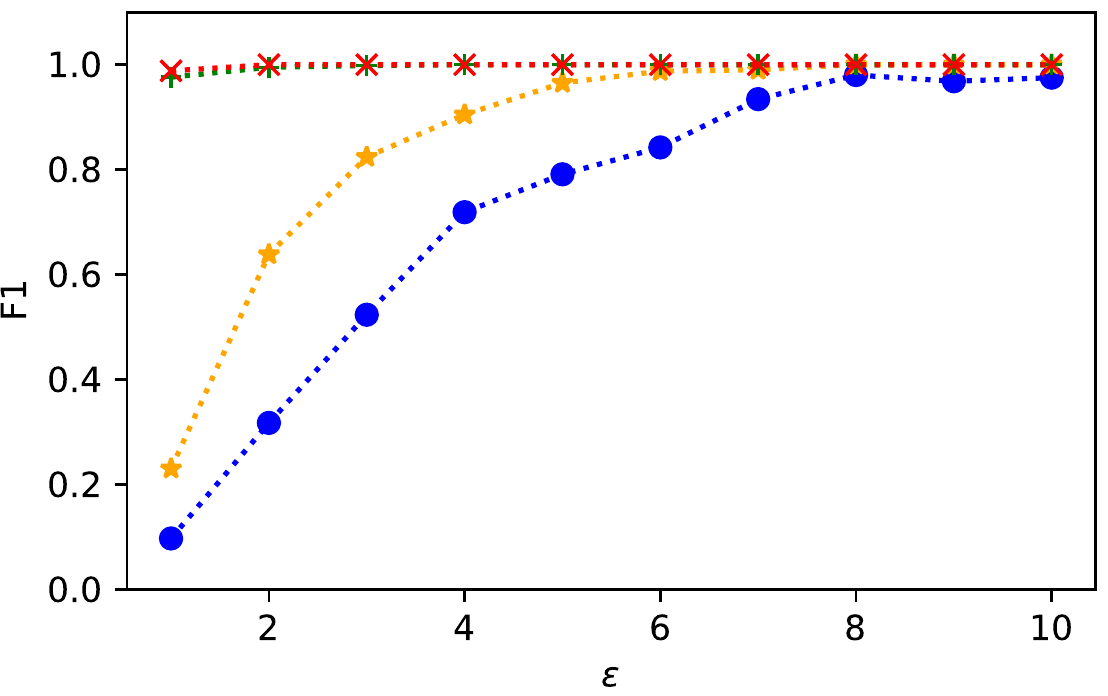}
			\caption{\texttt{HKT}}
		\end{subfigure}%
		~
		\begin{subfigure}[t]{0.33\textwidth}
			\centering
			\includegraphics[width=\linewidth]{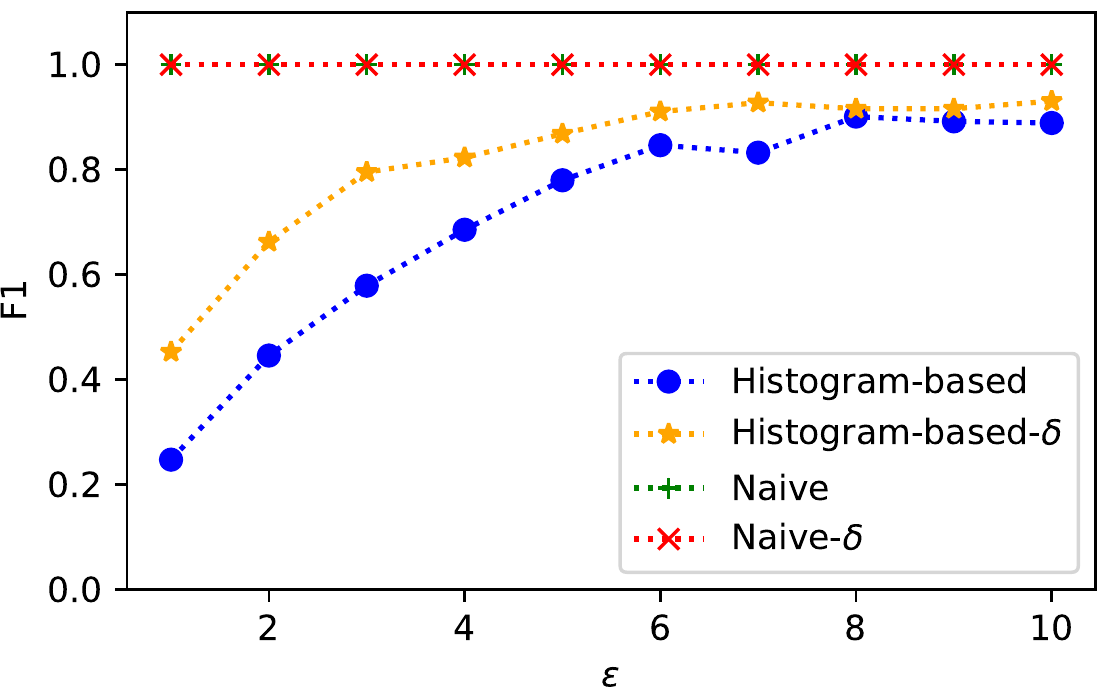}
			\caption{\texttt{HDY}}
		\end{subfigure}%
		\caption{$F_1$ result for different $\epsilon$ in all 3 monitored LANs}
		\label{fig:f1}
	\end{figure*}
	
	\noindent\textbf{Metrics.} In this experiment, we evaluate utility of our approaches using two metrics: true positive rate ($TPR$) and $F_1$ score.
	In particular, we consider $z^*[i]$, a noisy data point produced by our approach, to be a true positive ($TP$) if the anomaly detector classifies both $z^*[i]$ and $z[i]$ as an anomaly, where $z[i]$ represents the original non-privatized counterpart. $z^*[i]$ is a false positive ($FP$) if the anomaly detector finds an anomaly in $z^*[i]$ but not in $z[i]$. 
	A true negative ($TN$) and a false negative ($FN$) are also defined similarly. 
	
	Based on these definitions, $TPR$ and $F_1$ metrics can be formulated as:
	\begin{center}
		\[TPR = \frac{TP}{TP+FN}\]
		
		\[F_1 = \frac{TP}{TP+.5(FP+FN)}\]
	\end{center}
	
	A high value of $TPR$ implies that a high percentage of anomalies detected in the original data is also captured as an anomaly in the privatized data.
    On the other hand, a high value of $F_1$ implies relatively small values of $FP$ and $FN$ compared to $TP$. 
	~\\
	
	\noindent\textbf{Results.} Figure~\ref{fig:tpr} and~\ref{fig:f1} show the utility of our approaches evaluated using $TPR$ and $F_1$ metrics, respectively. First, we can see that $\epsilon$ does not affect utility of the na\"ive and na\"ive-$\delta$ approaches as both approaches still provide almost perfect utility scores in all monitored LANs. 
	
	On the other hand, the histogram-based and histogram-based-$\delta$ approaches yield low utility for small $\epsilon$. The utility scores then become higher as $\epsilon$ increases. For \texttt{HKT}, both approaches achieve a reasonable score of $>.75$ with $\epsilon = 5$. Meanwhile, $\epsilon$ must be set to $6$ in order to achieve the same utility score in \texttt{HDY}. \texttt{JPN} requires the highest $\epsilon$ ($=12$) in order for the histogram-based-$\delta$ approach to perform $75\%~TPR$.	
	
	Lastly, the results also confirm that the histogram-based-$\delta$ approach significantly outperforms the na\"ive-$\delta$ approach in terms of utility. Thus, we recommend to deploy the histogram-based-$\delta$ approach over the histogram-based approach when one needs to publish ARP-request data with user privacy protection (i.e., corresponding to the node-DP notion); whereas, if edge-DP is sufficient, the na\"ive approach is a more reasonable choice over the na\"ive-$\delta$ approach as the former provides a stronger privacy guarantee 
	while both approaches achieve the similar utility performance.~\\
	
	\noindent\textbf{Comparison with RMSE.} In most cases, the utility results from $TPR$ and $F_1$ metrics are consistent with the previous results measured using RMSE in Section~\ref{sec:rmse}. That is, a higher $\epsilon$ leads to higher utility with lower RMSE and higher $TPR$ and $F_1$. 
	On the other hand, an extremely low value of $\epsilon$ (e.g., $\epsilon = 1$) renders the output data useless as it can no longer be used to reveal anomalies due to its low $TPR$/$F_1$.
	There is, however, one exception: the na\"ive and na\"ive-$\delta$ approaches surprisingly can still attain high $TPR$ and $F_1$ utility despite low $\epsilon$. This indicates that such approaches are more robust to additive noises than other approaches.
	
	\section{Discussion}\label{sec:disc}
	
	\noindent\textbf{ARP Fields.} Our approaches take as input ARP-degree data, which in turn makes use of only 5 fields in ARP packets: SHA, SPA, THA, TPA, and OPER. 
	In this work, we choose to discard the rest of the ARP fields (i.e., Hardware Type/Length (HTYPE/HLEN) and Protocol Type/Length (PTYPE/PLEN)) from our analysis.
	This is because, in practice, these discarded fields usually have fixed values that contain neither sensitive information nor anything meaningful to our approaches. For instance, since ARP is only applicable to IPv4, the PLEN field is always set to the value of 4 indicating the size of an IPv4 address; or HTYPE usually contains the value of 1 representing the ubiquitous Ethernet hardware type. As these fields are generally constant for all ARP packets, their absence does not affect privacy or utility to our approaches. 
	~\\

        \noindent\textbf{DP Mechanisms.}
    In this work, we focus on releasing ARP-degrees in differentially private manners. Publishing degrees has sensitivity of $1$ (removing a user's ARP request alters the total ARP-degrees by $1$), which is small compared to the number of ARP requests sent by all users. Thus we choose the noise perturbation methods, namely the Laplace and the Gaussian mechanism, to privatize the ARP-degrees. Another well-known differential privacy mechanism is the randomized response, whose standard deviation is $O\left( \frac{\sqrt{N}}{\epsilon} \right)$~\cite{BeimelNO08}, which is worse than the standard deviation of the Laplace and Gaussian mechanism, which is $O\left( \frac{1}{\epsilon}\right)$. There are also differential privacy mechanisms based on data synthesis~\cite{tao2021benchmarking}. However, as anomaly detection algorithms look for ``spiking'' behaviors at a particular time interval, these data synthetic approaches, which try to replicate the distribution of the data as a whole, will not be able to retain the spikes as well as the perturbation mechanisms.~\\
	
	\noindent\textbf{Time Interval.} In our evaluation, we consider the time interval for ARP-data collection to be in a unit of a week.
	Albeit a bit long, this design choice is necessary as it allows us to incorporate all data (which spans for 30 weeks) into our analysis with higher utility rate and without losing too much privacy budget.
	
	To illustrate this point, we conduct a new experiment on the \texttt{JPN} network where we aggregate and process ARP data on a shorter period, i.e., every day instead of every week. 
	Compared to the original experiment, we have observed a drastic decrease in the utility rate for all our approaches.
	As an example, for the na\"ive approach with $\epsilon=4$, the RMSE has increased by a factor of $6$ (from 10 to 60), while the $TPR$ and $F_1$ score have reduced substantially from $1.0$ to $\approx0.6$.
	~\\
	
	\noindent\textbf{Utility Metrics.} We evaluate our approaches using two utility metrics: RMSE and Anomaly Detection Accuracy. 
	We select the former because it is one of the most common metrics for measuring utility from a DP mechanism~\cite{yang2017survey}.
    Intuitively, it tells us ``how far apart the privatized data is from the original data". Since an anomalous activity appears as an unusual value in the data, a privacy-preserving mechanism with small RMSE would not perturb \emph{that value} by much, allowing such activity to be detected from the privatized data. Besides RMSE, there are other similar metrics with the same purpose, e.g., Mean Absolute Error. Even though we do not include them in this work, we expect the results from such metrics to be in line with our current results.
	
	Nonetheless, the RMSE does not directly indicate the ``true" utility in this work since our end goal is to detect LAN anomalies, not minimize error rates.
	To this end, we choose to include Anomaly Detection Accuracy as our second metric. This metric realistically gives us an idea of how effective our approaches are when performing on a real-world LAN anomaly detector~\cite{matsufuji2019arp}.
	
	Finally, we do not consider other utility metrics that target different types of data publication. For example, $L_p$-Error~\cite{xiao2015protecting} and Hausdorff Distance~\cite{hua2015differentially} are geared towards measuring utility in location privacy protection. Also, information-theoretic metrics~\cite{lopuhaa2019information} require the input to be generated from a probability distribution, which is not the case in this work.

	\section{Conclusion}\label{sec:conclude}
	This paper presents four approaches to privately releasing ARP-request data that can later be used for identifying anomalies in LAN.
	We prove that the na\"ive approach satisfies edge-differential privacy, and thus provides privacy protection on the user-relationship level.
	On the other hand, the histogram-based approach can provide node-differential privacy, thus leaking no information about a presence of each individual user. 
	We also propose two alternatives, named na\"ive-$\delta$ and histogram-based-$\delta$, which require even smaller additive noises than their original counterparts in exchange for a small probability that the privacy guarantee will not hold.
	Feasibility of our approaches is demonstrated via real-world experiments in which we show that, with a reasonable privacy budget value, our approaches yield low errors ($<10$ in RMSE) and also preserve more than 75\% utility of detecting LAN anomalies.
	
	\ignore{
	\section*{Data Availability}
	The data used to support the findings of this study are available from the corresponding author upon request.
	
	\section*{Conflict of Interest}
	The authors declare that there are no conflicts of interest.
	
    \section*{Funding statement}
	The ASEAN IVO (\url{http://www.nict.go.jp/en/asean_ivo/index.html}) project, ASEAN-Wide Cyber-Security Research Testbed Project,
	was involved in the production of the contents of this work and
    financially supported by NICT (\url{http://www.nict.go.jp/en/index.html}). This work is also financially supported by Chiang Mai University, Thailand.

	\section*{Acknowledgments}
	The preliminary (and much shorter) version of this manuscript was published in IEEE International Conference on Electrical Engineering/Electronics, Computer, Telecommunications and Information Technology (ECTI-CON) 2021. The authors are grateful to the anonymous reviews of ECTI-CON 2021 for their helpful comments.
	}

\bibliographystyle{IEEEtranN}
\bibliography{degPriv}

\begin{thebibliography}{37}
\providecommand{\natexlab}[1]{#1}
\providecommand{\url}[1]{#1}
\csname url@samestyle\endcsname
\providecommand{\newblock}{\relax}
\providecommand{\bibinfo}[2]{#2}
\providecommand{\BIBentrySTDinterwordspacing}{\spaceskip=0pt\relax}
\providecommand{\BIBentryALTinterwordstretchfactor}{4}
\providecommand{\BIBentryALTinterwordspacing}{\spaceskip=\fontdimen2\font plus
\BIBentryALTinterwordstretchfactor\fontdimen3\font minus
  \fontdimen4\font\relax}
\providecommand{\BIBforeignlanguage}[2]{{%
\expandafter\ifx\csname l@#1\endcsname\relax
\typeout{** WARNING: IEEEtranN.bst: No hyphenation pattern has been}%
\typeout{** loaded for the language `#1'. Using the pattern for}%
\typeout{** the default language instead.}%
\else
\language=\csname l@#1\endcsname
\fi
#2}}
\providecommand{\BIBdecl}{\relax}
\BIBdecl

\bibitem[{IFP}(2018)]{vuln}
{IFP}, ``6 hacks sure to defeat your firewall (and how to prevent them),''
  \url{https://www.insightsforprofessionals.com/it/security/hacks-sure-to-defeat-your-firewall},
  2018.

\bibitem[Yan et~al.(2008)Yan, Zhang, and Ansari]{yan2008revealing}
W.~Yan, Z.~Zhang, and N.~Ansari, ``Revealing packed malware,'' \emph{IEEE
  Security \& Privacy}, vol.~6, no.~5, pp. 65--69, 2008.

\bibitem[Chapple(2021)]{ransom1}
M.~Chapple, ``The threat of ransomware still looms large over healthcare,''
  \url{https://healthtechmagazine.net/article/2021/06/threat-ransomware-still-looms-large-over-healthcare},
  2021.

\bibitem[Kern(2016)]{ransom2}
C.~Kern, ``95
  filtering,''
  \url{https://www.varinsights.com/doc/of-ransomware-bypass-firewalls-email-filtering-0001},
  2016.

\bibitem[{Valeo Networks}(2021)]{iotmalware}
{Valeo Networks}, ``How is the internet of things (iot) being impacted by
  malware?''
  \url{https://www.valeonetworks.com/how-is-the-internet-of-things-iot-being-impacted-by-malware/},
  2021.

\bibitem[Matsufuji et~al.(2019)Matsufuji, Kobayashi, Esaki, and
  Ochiai]{matsufuji2019arp}
K.~Matsufuji, S.~Kobayashi, H.~Esaki, and H.~Ochiai, ``Arp request trend
  fitting for detecting malicious activity in lan,'' in \emph{International
  Conference on Ubiquitous Information Management and Communication (ICUIMC)},
  2019.

\bibitem[Yasami et~al.(2007)Yasami, Farahmand, and Zargari]{yasami2007arp}
Y.~Yasami, M.~Farahmand, and V.~Zargari, ``An arp-based anomaly detection
  algorithm using hidden markov model in enterprise networks,'' in
  \emph{International Conference on Systems and Networks Communications
  (ICSNC)}.\hskip 1em plus 0.5em minus 0.4em\relax IEEE, 2007.

\bibitem[Ren et~al.(2019)Ren, Xu, Wang, Yi, Huang, Kou, Xing, Yang, Tong, and
  Zhang]{ren2019time}
H.~Ren, B.~Xu, Y.~Wang, C.~Yi, C.~Huang, X.~Kou, T.~Xing, M.~Yang, J.~Tong, and
  Q.~Zhang, ``Time-series anomaly detection service at microsoft,'' in
  \emph{ACM SIGKDD International Conference on Knowledge Discovery \& Data
  Mining}, 2019.

\bibitem[Mobilio et~al.(2019)Mobilio, Orr{\`u}, Riganelli, Tundo, and
  Mariani]{mobilio2019anomaly}
M.~Mobilio, M.~Orr{\`u}, O.~Riganelli, A.~Tundo, and L.~Mariani, ``Anomaly
  detection as-a-service,'' in \emph{IEEE International Symposium on Software
  Reliability Engineering Workshops (ISSREW)}, 2019.

\bibitem[Yao et~al.(2017)Yao, Shu, Cheng, and Stolfo]{yao2017anomaly}
D.~Yao, X.~Shu, L.~Cheng, and S.~J. Stolfo, ``Anomaly detection as a service:
  challenges, advances, and opportunities,'' \emph{Synthesis Lectures on
  Information Security, Privacy, and Trust}, vol.~9, no.~3, pp. 1--173, 2017.

\bibitem[{Microsoft Azure}(2020)]{microsoftanomaly}
{Microsoft Azure}, ``Anomaly detector,''
  \url{https://azure.microsoft.com/en-us/services/cognitive-services/anomaly-detector/},
  2020.

\bibitem[{TIBCO Software}(2021)]{tibco}
{TIBCO Software}, ``Anomaly detection | tibco software,''
  \url{https://www.tibco.com/solutions/anomaly-detection}, 2021.

\bibitem[Hu et~al.(2016)Hu, Lin, and Li]{hu2016relationship}
J.~Hu, C.~Lin, and X.~Li, ``Relationship privacy leakage in network traffics,''
  in \emph{2016 25th International Conference on Computer Communication and
  Networks (ICCCN)}.\hskip 1em plus 0.5em minus 0.4em\relax IEEE, 2016, pp.
  1--9.

\bibitem[Srivatsa and Hicks(2012)]{srivatsa2012deanonymizing}
M.~Srivatsa and M.~Hicks, ``Deanonymizing mobility traces: Using social network
  as a side-channel,'' in \emph{ACM Conference on Computer and Communications
  Security}, 2012.

\bibitem[Takbiri et~al.(2018)Takbiri, Houmansadr, Goeckel, and
  Pishro-Nik]{takbiri2018privacy}
N.~Takbiri, A.~Houmansadr, D.~L. Goeckel, and H.~Pishro-Nik, ``Privacy against
  statistical matching: Inter-user correlation,'' in \emph{IEEE International
  Symposium on Information Theory (ISIT)}, 2018.

\bibitem[Dwork et~al.(2006)Dwork, Mcsherry, Nissim, and Smith]{Dwork2006}
C.~Dwork, F.~Mcsherry, K.~Nissim, and A.~Smith, ``Calibrating noise to
  sensitivity in private data analysis,'' in \emph{Theory of Cryptography
  Conference (TCC)}, 2006.

\bibitem[McSherry and Mahajan(2010)]{mcsherry2010differentially}
F.~McSherry and R.~Mahajan, ``Differentially-private network trace analysis,''
  \emph{ACM SIGCOMM Computer Communication Review}, vol.~40, no.~4, pp.
  123--134, 2010.

\bibitem[McSherry(2009)]{mcsherry2009privacy}
F.~D. McSherry, ``Privacy integrated queries: an extensible platform for
  privacy-preserving data analysis,'' in \emph{International Conference on
  Management of data}, 2009.

\bibitem[Fan et~al.(2014)Fan, Bonomi, Xiong, and Sunderam]{fan2014monitoring}
L.~Fan, L.~Bonomi, L.~Xiong, and V.~Sunderam, ``Monitoring web browsing
  behavior with differential privacy,'' in \emph{International Conference on
  World Wide Web (WWW)}, 2014.

\bibitem[Wang et~al.(2016)Wang, Zhang, Lu, Wang, Qin, and Ren]{wang2016real}
Q.~Wang, Y.~Zhang, X.~Lu, Z.~Wang, Z.~Qin, and K.~Ren, ``Real-time and
  spatio-temporal crowd-sourced social network data publishing with
  differential privacy,'' \emph{IEEE Transactions on Dependable and Secure
  Computing}, vol.~15, no.~4, pp. 591--606, 2016.

\bibitem[Dankar and El~Emam(2013)]{dankar2013practicing}
F.~K. Dankar and K.~El~Emam, ``Practicing differential privacy in health care:
  A review.'' \emph{Trans. Data Priv.}, vol.~6, no.~1, pp. 35--67, 2013.

\bibitem[Fan and Xiong(2013)]{fan2013differentially}
L.~Fan and L.~Xiong, ``Differentially private anomaly detection with a case
  study on epidemic outbreak detection,'' in \emph{International Conference on
  Data Mining Workshops (ICDMW)}, 2013.

\bibitem[Zhang et~al.(2019)Zhang, Esaki, and Ochiai]{zhang2019unveiling}
Z.~Zhang, H.~Esaki, and H.~Ochiai, ``Unveiling malicious activities in lan with
  honeypot,'' in \emph{International Conference on Information Technology
  (InCIT)}, 2019.

\bibitem[Yeo et~al.(2004)Yeo, Youssef, and Agrawala]{yeo2004framework}
J.~Yeo, M.~Youssef, and A.~Agrawala, ``A framework for wireless lan monitoring
  and its applications,'' in \emph{Proceedings of the 3rd ACM workshop on
  Wireless security}, 2004.

\bibitem[Whyte et~al.(2005)Whyte, van Oorschot, and Kranakis]{whyte2005arp}
D.~Whyte, P.~van Oorschot, and E.~Kranakis, ``Arp-based detection of scanning
  worms within an enterprise network,'' in \emph{Proceedings of the Annual
  Computer Security Applications Conference (ACSAC)}, 2005.

\bibitem[Farahmand et~al.(2006)Farahmand, Azarfar, Jafari, and
  Zargari]{farahmand2006multivariate}
M.~Farahmand, A.~Azarfar, A.~Jafari, and V.~Zargari, ``A multivariate adaptive
  method for detecting arp anomaly in local area networks,'' in
  \emph{International Conference on Systems and Networks Communications
  (ICSNC)}, 2006.

\bibitem[Bun and Steinke(2016)]{Bun2016}
M.~Bun and T.~Steinke, ``Concentrated differential privacy: Simplifications,
  extensions, and lower bounds,'' in \emph{Theory of Cryptography Conference
  (TCC)}, 2016.

\bibitem[Dwork and Roth(2014)]{Dwork2014}
C.~Dwork and A.~Roth, ``The algorithmic foundations of differential privacy,''
  \emph{Foundations and Trends® in Theoretical Computer Science}, vol.~9, no.
  3–4, pp. 211--407, 2014.

\bibitem[Hay et~al.(2009)Hay, Li, Miklau, and Jensen]{Hay2009}
\BIBentryALTinterwordspacing
M.~Hay, C.~Li, G.~Miklau, and D.~D. Jensen, ``Accurate estimation of the degree
  distribution of private networks,'' in \emph{{ICDM} 2009, The Ninth {IEEE}
  International Conference on Data Mining, Miami, Florida, USA, 6-9 December
  2009}, W.~Wang, H.~Kargupta, S.~Ranka, P.~S. Yu, and X.~Wu, Eds.\hskip 1em
  plus 0.5em minus 0.4em\relax {IEEE} Computer Society, 2009, pp. 169--178.
  [Online]. Available: \url{https://doi.org/10.1109/ICDM.2009.11}
\BIBentrySTDinterwordspacing

\bibitem[Mironov(2012)]{Ilya2012}
I.~Mironov, ``On significance of the least significant bits for differential
  privacy,'' in \emph{ACM Conference on Computer and Communications Security},
  2012.

\bibitem[Montgomery and Mastrangelo(1991)]{Montgomery1991}
\BIBentryALTinterwordspacing
D.~C. Montgomery and C.~M. Mastrangelo, ``Some statistical process control
  methods for autocorrelated data,'' \emph{Journal of Quality Technology},
  vol.~23, no.~3, pp. 179--193, Jul. 1991. [Online]. Available:
  \url{https://doi.org/10.1080/00224065.1991.11979321}
\BIBentrySTDinterwordspacing

\bibitem[Beimel et~al.(2008)Beimel, Nissim, and Omri]{BeimelNO08}
A.~Beimel, K.~Nissim, and E.~Omri, ``Distributed private data analysis:
  Simultaneously solving how and what,'' in \emph{Advances in Cryptology -
  {CRYPTO} 2008, 28th Annual International Cryptology Conference, Santa
  Barbara, CA, USA, August 17-21, 2008. Proceedings}, ser. Lecture Notes in
  Computer Science, D.~A. Wagner, Ed., vol. 5157.\hskip 1em plus 0.5em minus
  0.4em\relax Springer, 2008, pp. 451--468.

\bibitem[Tao et~al.(2021)Tao, McKenna, Hay, Machanavajjhala, and
  Miklau]{tao2021benchmarking}
Y.~Tao, R.~McKenna, M.~Hay, A.~Machanavajjhala, and G.~Miklau, ``Benchmarking
  differentially private synthetic data generation algorithms,'' \emph{arXiv
  preprint arXiv:2112.09238}, 2021.

\bibitem[Yang et~al.(2017)Yang, Wang, Ren, and Yu]{yang2017survey}
X.~Yang, T.~Wang, X.~Ren, and W.~Yu, ``Survey on improving data utility in
  differentially private sequential data publishing,'' \emph{IEEE Transactions
  on Big Data}, 2017.

\bibitem[Xiao and Xiong(2015)]{xiao2015protecting}
Y.~Xiao and L.~Xiong, ``Protecting locations with differential privacy under
  temporal correlations,'' in \emph{Proceedings of the 22nd ACM SIGSAC
  Conference on Computer and Communications Security}, 2015, pp. 1298--1309.

\bibitem[Hua et~al.(2015)Hua, Gao, and Zhong]{hua2015differentially}
J.~Hua, Y.~Gao, and S.~Zhong, ``Differentially private publication of general
  time-serial trajectory data,'' in \emph{2015 IEEE Conference on Computer
  Communications (INFOCOM)}.\hskip 1em plus 0.5em minus 0.4em\relax IEEE, 2015,
  pp. 549--557.

\bibitem[Lopuha{\"a}-Zwakenberg et~al.(2019)Lopuha{\"a}-Zwakenberg,
  {\v{S}}kori{\'c}, and Li]{lopuhaa2019information}
M.~Lopuha{\"a}-Zwakenberg, B.~{\v{S}}kori{\'c}, and N.~Li,
  ``Information-theoretic metrics for local differential privacy protocols,''
  \emph{arXiv preprint arXiv:1910.07826}, 2019.

\end{thebibliography}

    \end{document}